\documentclass[preprint, 3p]{elsarticle}

\usepackage{subcaption}
\usepackage[utf8]{inputenc}
\usepackage{tikz-cd} 
\usepackage{amsmath}
\usepackage{amsthm}
\usepackage{amsfonts}
\usepackage{mathdots} 
\usepackage{bbm}
 \usepackage{amssymb}
 \usepackage{amsthm}
 \usepackage{amsmath}
\usepackage{natbib}
 \usepackage{tikz}
 \usepackage{graphicx}
 \usepackage{epstopdf}
 \usepackage{xurl}
 \usepackage{hyperref}
 \usepackage{multirow}
 \usetikzlibrary{patterns}
\usepackage{booktabs}
\usepackage{lineno}

\newtheorem{theorem}{Theorem}[section]

\definecolor{myYellow}{rgb}{1.0, 0.87, 0.0}
\definecolor{myYtable}{rgb}{1.0, 0.75, 0.0}
\definecolor{umass}{rgb}{0.004, 0.3215, 0.5568}
\definecolor{etonblue}{rgb}{0.59, 0.78, 0.64}
 \definecolor{mycolor}{rgb}{0.12, 0.3, 0.17}
\definecolor{myblue}{rgb}{0.03, 0.27, 0.49}
\definecolor{mediumjunglegreen}{rgb}{0.11, 0.21, 0.18}
\definecolor{asparagus}{rgb}{0.53, 0.66, 0.42}
\definecolor{goldenpoppy}{rgb}{0.99, 0.76, 0.0}
\definecolor{lincolngreen}{rgb}{0.11, 0.35, 0.02}
\definecolor{red(ncs)}{rgb}{0.77, 0.01, 0.2}
\definecolor{darkorange}{rgb}{1.0, 0.55, 0.0}
\definecolor{safetyorange(blazeorange)}{rgb}{1.0, 0.4, 0.0}
\definecolor{smokeytopaz}{rgb}{0.58, 0.25, 0.03}
\definecolor{brightgreen}{rgb}{0.4, 1.0, 0.0}
\definecolor{tue}{rgb}{0.82, 0.14, 0.14}

\newcommand{\EE}{{\mathbb{E}}}
\newcommand*\diff{\mathop{}\!\mathrm{d}}

\newcommand{\ja}{\textcolor{black}}

\journal{arXiv}

\usepackage{xcolor}

\usepackage{url}

\begin{document}

\begin{frontmatter}

\title{Rank-based concordance for zero-inflated data: New representations, estimators, and sharp bounds}

\author{Jasper Arends\footnote{Department of Mathematics and Computer Science, Eindhoven University of Technology, Groene Loper 5, 5612 AZ, Eindhoven, The Netherlands, E-mail: j.r.m.arends@tue.nl},\
       Guanjie Lyu\footnote{ 
		Department of Mathematics and Statistics, University of Windsor, Canada, E-mail: lvg@\,uwindsor.ca} ,\
	Mhamed Mesfioui\footnote{Department of Mathematics and Computer Science,
University of Quebec in Trois-Riv\`{e}res, Canada, E-mail: mhamed.mesfioui@uqtr.ca}, \
     Elisa Perrone\footnote{Department of Mathematics and Computer Science, Eindhoven University of Technology, Groene Loper 5, 5612 AZ, Eindhoven, The Netherlands, E-mail: e.perrone@tue.nl}
	 \ and\
	Julien Trufin\footnote{Department of Mathematics, Universit\'e Libre de Bruxelles (ULB), Belgium,  E-mail: julien.trufin@ulb.be}
}

	\begin{abstract}
 Quantifying concordance between two random variables is crucial in applications. Traditional estimation techniques for commonly used concordance measures, such as Gini's gamma or Spearman's rho, often fail when data contain ties. 
This is particularly problematic for zero-inflated data, characterized by a combination of discrete mass in zero and a continuous component, which frequently appear in insurance, weather forecasting, and biomedical applications. This study provides a new formulation of Gini’s gamma and Spearman’s footrule, two rank-based concordance measures that incorporate absolute rank differences, tailored to zero-inflated continuous distributions. Along the way, we correct an expression of Spearman's rho for zero-inflated data previously presented in the literature. The best-possible upper and lower bounds for these measures in zero-inflated continuous settings are established, making the estimators useful and interpretable in practice. 
We pair our theoretical results with simulations and two real-life applications in insurance and weather forecasting, respectively. Our results illustrate the impact of zero inflation on dependence estimation, emphasizing the benefits of appropriately adjusted zero-inflated measures.
\end{abstract}

\end{frontmatter}

	\section{Introduction}\label{sec:introduction}
 Measuring association plays a crucial role in statistical analysis, including in applications involving zero-inflated continuous distributions, i.e., continuous distributions with an extra probability mass in zero. Many real-world datasets exhibit a combination of discrete and continuous components, such as those encountered in finance, insurance, weather forecasting, and biomedical studies, where observations contain a substantial proportion of zeros. Traditional association measures, including Kendall’s tau and Spearman’s rho \cite{Nelsen2006}, are widely used to assess concordance in purely continuous settings. 
 However, their direct application to zero-inflated continuous settings can result in misleading conclusions due to zero-inflated structures. To accurately capture dependence in such datasets, refined measures are needed that effectively address the mixed nature of zero-inflated continuous distributions.
 

Recent studies, grounded in the framework of copula theory \cite{Nelsen2006}, have extended classical measures of association to account for excess zeros. 
Specifically, modified versions of Kendall’s tau and Spearman’s rho that integrate zero inflation into their formulations and their attainable range are presented in the literature \cite{Denuit2017bounds, Mesfioui2022rho, perrone2023,Pimentel2009kendall, Pimentel2015kendall}. 
Expanding on these developments, the present study focuses on Gini’s gamma and Spearman’s footrule, two alternative measures of concordance that offer distinct advantages in rank-based association analysis \cite{Nelsen2006}. Unlike Kendall’s tau and Spearman’s rho, which emphasize global monotonicity, Gini’s gamma and Spearman’s footrule provide a more nuanced characterization of dependence by incorporating absolute rank differences.

In this paper, we propose a new formulation of Gini’s gamma and Spearman’s footrule tailored to zero-inflated continuous cases. Along the way, we also provide a correction for Spearman's rho in such settings, addressing inaccuracies identified in previous formulations. 
In the realm of classical dependence measures, the achievable bounds for zero-inflated continuous distributions are conditional upon the Fréchet-Hoeffding limits, which delineate the strongest and weakest conceivable dependence structures based on specified marginal distributions. 
We expand this theoretical framework and establish precise bounds that encapsulate the range of values that these concordance measures can assume, thereby enabling a more interpretable assessment of the strength of association in these scenarios.

The remainder of the paper is structured as follows. \hyperref[sec:definition]{Section~\ref{sec:definition}} introduces the notation and define Gini’s gamma and Spearman’s footrule as concordance measures. In~\hyperref[sec:results]{Section~\ref{sec:results}}, we focus on the case of zero-inflated continuous distributions and present our theoretical results. Namely, we derive the analytical expressions of Gini's gamma and Spearman's footrule in this setting and provide a correction to Spearman’s rho, addressing inaccuracies in previous formulations. Additionally, we establish the best-possible upper and lower bounds for these measures. \hyperref[sec:simulation]{Section~\ref{sec:simulation}} and ~\hyperref[sec:application]{Section~\ref{sec:application}} illustrate the applicability of the theoretical results for simulated and real-world data, respectively. Finally, we conclude with a discussion of our findings and potential directions for future research in ~\hyperref[sec:conclusion]{Section~\ref{sec:conclusion}}.

	\section{Background}\label{sec:definition}
		
	   Let $(X, Y)$ be a pair of zero-inflated continuous random variables with joint distribution function $J$ and marginal distributions $F$ and $G$, respectively. Similarly as in~\cite{Denuit2017bounds} and~\cite{Mesfioui2022rho}, let $\mathbbm 1(\cdot)$ denote the indicator function and let $p_1$ $(p_2)$ represent the probability mass at zero for $X$ $(Y)$. The distributions of $X$ and $Y$ can then be written as
	\begin{align}\label{eq:distr}
		F(s)&=p_1\mathbbm{1}(s=0)+\left[p_1+(1-p_1)\tilde F(s) \right]\mathbbm{1}(s>0), \notag\\  G(s)&=p_2\mathbbm{1}(s=0)+\left[p_2+(1-p_2)\tilde G(s) \right]\mathbbm{1}(s>0),
	\end{align}
    where $\tilde F$ and $\tilde G$ denote the distribution functions of $X|X>0$ and $Y|Y>0$, respectively. For any $s\geq 0$, let us define the left-continuous versions of the marginal distributions as $F(s^-)=\text{P}(X<s)$ and $G(s^-)=\text{P}(Y<s)$. Then, it follows that
	\begin{eqnarray*}
		F(s^-)=\left[p_1+(1-p_1)\tilde F(s) \right]\mathbbm{1}(s>0), \quad G(s^-)=\left[p_2+(1-p_2)\tilde G(s) \right]\mathbbm{1}(s>0).
	\end{eqnarray*}

	We recall that a copula $C$ is a cumulative joint distribution function with uniform margins in $[0, 1]$ \cite{DurSem15, Joe20141, Nelsen2006}. According to Sklar's theorem \cite{sklar_59}, any joint distribution can be expressed as a copula $C$ that takes as arguments the marginal distributions of the random vector, that is $P[X \leq x, Y \leq y] = C(F(x), G(y))$.
	  Within the copula framework, let $(X_\star, Y_\star)$ be another pair independent of $(X, Y)$, where $X_\star$ and $Y_\star$ are governed by the joint distribution function $J_\star$ with the same marginal distributions as $X$ and $Y$. From~\cite[Chapter 5]{Nelsen2006}, the concordance function is defined as
	\begin{equation}\label{eq:2025-01-26, 3:41PM}
		Q(J, J_\star)=P\left((X-X_\star)(Y-Y_\star)>0\right)-P\left((X-X_\star)(Y-Y_\star)<0\right).
	\end{equation}

For continuous random variables, Gini's gamma ($\gamma$) and Spearman's footrule ($\phi$) are commonly defined as follows
\begin{align}\label{eq:defin_traditional}
    \gamma = Q(J, M) + Q(J, W), \quad \phi = \frac{1}{2} \left[3 Q(J, M) - 1\right],
\end{align}
where $M(x, y)=\min\{F(x), G(y)\}$ and $W(x, y)=\max\{F(x)+G(y)-1, 0\}$ denote the Fr\'echet-Hoeffding bounds of copulas. Although they are equal to zero under independence, this may not necessarily hold for random variables with a discrete component. To account for this, we instead define Gini's gamma and Spearman's footrule as
	\begin{align}\label{eq:gamma_footrule}
		\gamma=Q(J,  M)+Q(J, W) - Q(\Pi, M) - Q(\Pi, W), \quad \phi=\frac{3}{2}\left[Q(J, M) - Q(\Pi, M)\right].
	\end{align}
These definitions simplify to the traditional definitions given in Eq.~\eqref{eq:defin_traditional} in fully continuous settings and preserve the property of zero under independence in zero-inflated continuous random variables.

Based on Eq.~\eqref{eq:2025-01-26, 3:41PM}, various other measures of association can be introduced for zero-inflated data. For example, setting $J_\star\equiv J$ yields
$\tau=Q(J, J)$, which corresponds to Kendall's $\tau$ for zero-inflated continuous distributions. Alternatively, when $J_\star \equiv \Pi$, where $\Pi$ denotes the independence copula, we obtain $\rho=3Q(J, \Pi)$, which defines Spearman's $\rho$. Both measures and their estimates were developed in~\cite{Pimentel2009kendall}. New formulations and, as a consequence, estimators for Gini's gamma and Spearman's footrule for zero-inflated cases can be derived along the same lines as discussed in the next section. Before presenting our results, we introduce the following notation, which we use in the remainder of the paper:
\begin{align}\label{eq:prob0011}
	p_{00}&=P(X=0, Y=0), \quad p_{01}=P(X=0, Y>0),  \notag\\
	p_{10}&=P(X>0, Y=0), \quad p_{11}=P(X>0, Y>0),
\end{align}
and for random pairs $(X, Y)$ and $(X_\star, Y_\star)$, 
\begin{align}\label{eq:probPiCond}
    \pi_{1, J_\star} &= P(\{X\ |\ X > 0, Y > 0\} > \{X_\star\ |\ X_\star > 0, Y_\star > 0\}),\notag\\
    \pi_{2, J_\star} &= P(\{Y\ |\ X > 0, Y > 0\} > \{Y_\star\ |\ X_\star > 0, Y_\star > 0\}),\notag\\
	\pi_{3, J_\star}&=P(\{Y|X>0, Y>0\} > \{Y_\star|X_\star=0, Y_\star>0\}),\notag\\
	\pi_{4, J_\star}&=P(\{X|X>0, Y>0\} > \{X_\star|X_\star>0, Y_\star=0\}).
\end{align}

\section{Results}
\label{sec:results}
\subsection{Gini's gamma and Spearman's footrule for zero-inflated continuous settings}
In this section, we present new expressions and corresponding estimators for Gini's gamma and Spearman's footrule tailored to zero-inflated continuous data. 
Our expressions are based on the intuition that the concordance function in Eq.~\eqref{eq:2025-01-26, 3:41PM} can be split into two components, i.e., one part at zero and the other away from zero. A similar strategy has also been used in \cite{perrone2023,Pimentel2015kendall,Pimentel2009kendall} for estimating Kendall's tau, and it is proved to be useful in our context as well.

In fact, our main results are new derivations, for zero-inflated continuous distributions, of the concordance functions with respect to the upper and lower Fr\'echet-Hoeffding copulas, i.e., $Q(J, M)$ and $Q(J, W)$.
These derivations can then be put together in Eq.~\eqref{eq:gamma_footrule} to obtain new expressions of these concordance measures.
 Before stating our main findings in Theorem~\ref{thm:conc_upper} and Theorem~\ref{thm:conc_lower}, we provide a brief explanation of the ideas and notation used.

Without loss of generality, we assume that $p_1 \leq p_2$. We first consider a pair of zero-inflated random variables $(X_\star, Y_\star)$ joined through the upper Fr\'echet-Hoeffding copula $M$. As noted in \cite{Denuit2017bounds}, in this case $X_\star$ given $\{X_\star > 0, Y_\star > 0\}$ or $\{X_\star > 0, Y_\star = 0\}$ cannot attain all positive values.
Specifically, as depicted in Figure~\ref{fig:copula_M}, we have that $$F(\{X_\star\ |\ X_\star > 0, Y_\star > 0\}) > p_2, \quad F(\{X_\star\ |\ X_\star > 0, Y_\star = 0\}) \leq p_2.$$
Therefore, a new expression for $Q(J, M)$ can be derived by separating the cases where $X_1, Y_1$ are zero and/or positive, and satisfy 
\begin{itemize}
\item[$\rm (I)$] $0 < X_1 \leq F^{-1}(p_2)$;
\item[$\rm (II)$] $X_2 > F^{-1}(p_2)$.
\end{itemize} 
We examine the probabilities outlined in Eq.~\eqref{eq:prob0011} and Eq.~\eqref{eq:probPiCond}, where $X$ fulfills either condition $\rm (I)$ or $\rm (II)$ as denoted by a superscript. For instance, $p_{10}^{\rm (I)} = P\bigl(X > 0, Y = 0, X \leq F^{-1}(p_2)\bigr)$. By analyzing each scenario separately, we derive the representation of $Q(J, M)$ as provided below.

\begin{figure}[t]
    \centering
    \begin{subfigure}[t]{0.45\textwidth}
        \centering
        \includegraphics[height=100px]{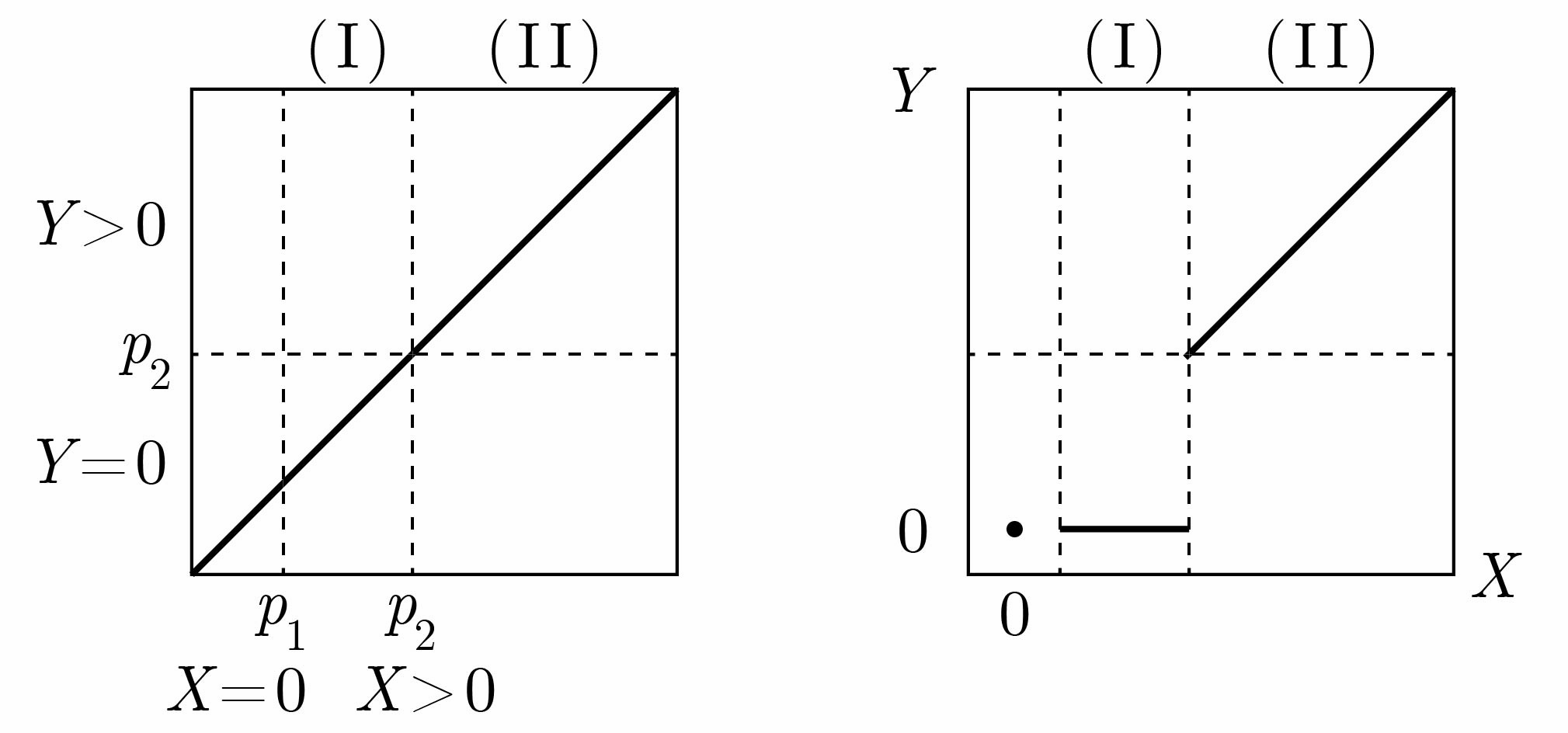}
        \caption{For the upper Fr\'echet-Hoeffding copula bound, the region $\{X > 0\}$ can be partitioned into $\rm (I)$ and $\rm (II)$.}
        \label{fig:copula_M}
    \end{subfigure}
    \hfill
    \begin{subfigure}[t]{0.45\textwidth}
        \centering
        \includegraphics[height=100px]{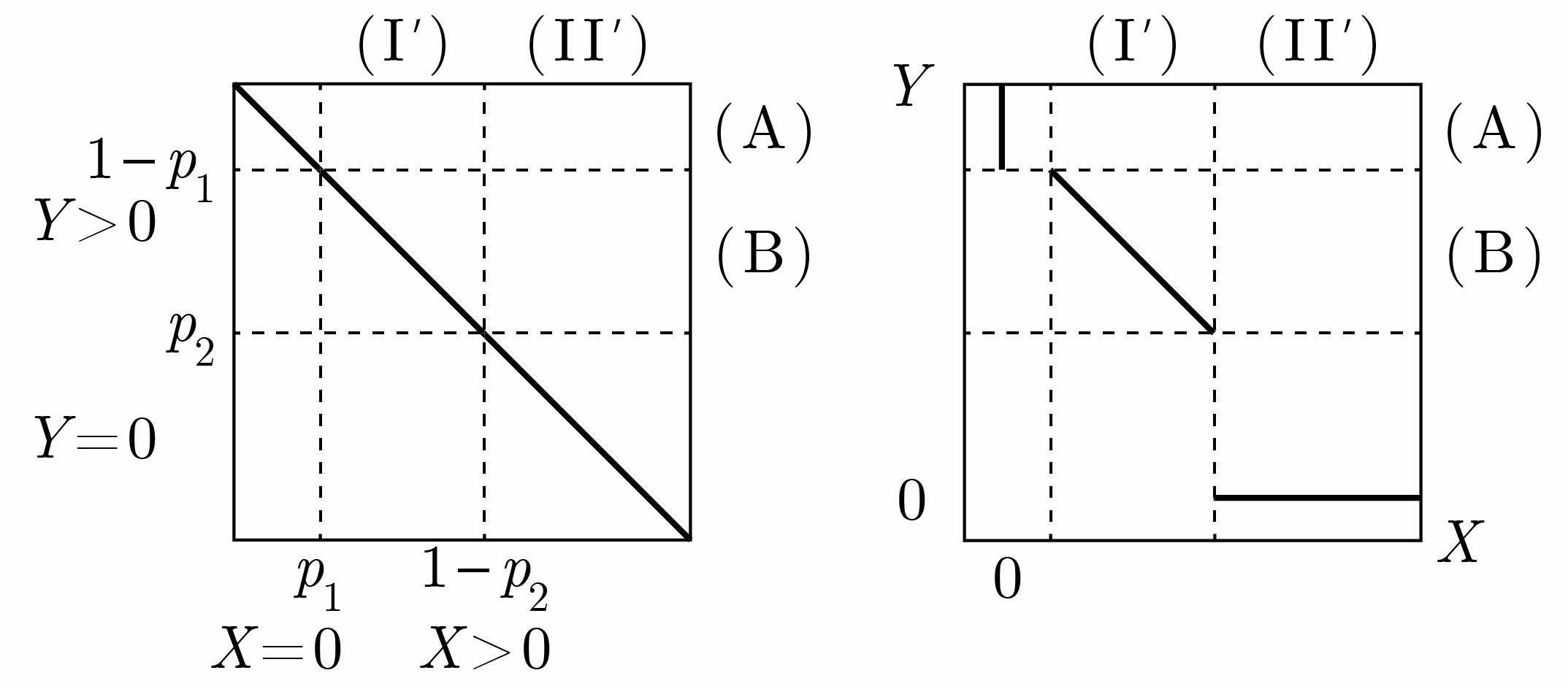}
        \caption{For the lower Fr\'echet-Hoeffding copula bound, the regions $\{X > 0\}$ and $\{Y > 0\}$ can be partitioned into $\rm (I' )$, $\rm (II' )$ and $\rm (A)$, $\rm (B)$ respectively.}
        \label{fig:copula_W}
    \end{subfigure}
    \caption{Copulas (left) and corresponding sample support (right) for zero-inflated random variables $(X, Y)$ joined through the upper and lower Fr\'echet-Hoeffding copula bounds, (a) and (b) respectively, satisfying $p_1 \leq p_2$ and $p_1 + p_2 \leq 1$.}
\end{figure}


\begin{theorem}\label{thm:conc_upper}
    Assuming $p_1 \leq p_2$, the concordance function relative to the upper Fr\'echet-Hoeffding copula is expressed as
    \begin{align*}
        Q(J, M) &= (1 - p_2) p_{11}^{\rm (II)} \left(4 C_M - 1\right) + (p_2 - p_1) \bigl(2 \pi_{4, M}^{\rm (I)} - 1\bigr) \\
        &\qquad + (1 - p_2) \bigl(p_{00} + p_{10}^{\rm (I)}\bigr) + p_{11} p_1 + (p_2 - p_1) \bigl(p_{11}^{\rm (II)} - p_{01}\bigr).
    \end{align*}
    $C_M$ represents the concordance where both $X$ and $Y$ are positive and limited to $\rm (I)$, which is specified as
    \begin{align*}
        C_M = P\bigl(X > X_\star, Y > Y_\star \,\bigm|\, X > 0, Y > 0, X \leq F^{-1}(p_2) \bigr).
    \end{align*}
\end{theorem}
\begin{proof}
	The proof is provided in the~\hyperref[app]{Appendix}. 
\end{proof}

In a similar manner, we now examine the scenario where the pair $(X_\star, Y_\star)$ is linked via the lower Fr\'echet-Hoeffding copula with the condition $p_1 + p_2 \leq 1$. \cite{Denuit2017bounds} highlighted that $\{X_\star\ |\ X_\star > 0\}$ cannot cover the entire range of positive values when $Y_\star = 0$ or $Y_\star > 0$, and the reverse is also true. These specific intervals are depicted in Figure~\ref{fig:copula_W}.
Thus, it is essential to identify additional scenarios beyond $X$ and/or $Y$ merely being non-negative or zero, and to further divide $\{X > 0\}$ and $\{Y > 0\}$ into these specific cases, which must then be combined:
\begin{itemize}
\item[${\rm (I')}$] $0 < X \leq F^{-1}(1 - p_2)$;
\item[${\rm (II')}$] $F^{-1}(1 - p_2) < X$;
\item[${\rm (A)}$] $Y > G^{-1}(1 - p_1)$;
\item[${\rm (B)}$] $0 < Y \leq G^{-1}(1 - p_1)$.
\end{itemize}
We revisit the probabilities outlined in Eq.~\eqref{eq:prob0011} and Eq.~\eqref{eq:probPiCond}, which are applicable to $X$ and $Y$ under specified conditions. These are denoted using a superscript, for instance, $p_{11}^{\rm (AI')} = P\bigl(0 < X \leq F^{-1}(1 - p_2), Y > G^{-1}(1 - p_1)\bigr)$.
Ultimately, when the joint distribution of $(X_\star, Y_\star)$ corresponds to the lower Fr\'echet-Hoeffding copula with $p_1 + p_2 > 1$, \cite{Denuit2017bounds} noted that $P(X_\star > 0, Y_\star > 0) = 0$. Consequently, further distinctions between cases are unnecessary.

\begin{theorem}\label{thm:conc_lower}
    Suppose that $p_1 + p_2 \leq 1$, then the concordance function with respect to the lower Fr\'echet-Hoeffding copula is given by
    \begin{align*}
        Q(J, W) &= (1 - p_1 - p_2) p_{11}^{\rm (BI')} (4 C_W - 1 + 4D_1) + 2 D_2 - p_1 (1 - p_1) - p_2 (1 - p_2).
    \end{align*}
    Here, $C_W$ represents the concordance of $X$ and $Y$ being positive and is defined as
    \begin{align*}
        C_W = P\bigl(X > X_\star, Y > Y_\star \,|\, 0 < X \leq F^{-1}(1 - p_2), 0 < Y \leq G^{-1}(1 - p_1)\bigr).
    \end{align*}
    Moreover, $D_1$ and $D_2$ measure the extent to which $X$ and $Y$ can exceed $X_\star$ and $Y_\star$ respectively, provided that they are positive, and they are specified by
    \begin{align*}
        D_1 &= p_{11}^{\rm (AI')} \pi_{1, W}^{\rm (AI')} + p_{11}^{\rm (BII')} \pi_{2, W}^{\rm (BII')} + p_{11}^{\rm (AII')}, \qquad D_2 = p_1 p_{11}^{\rm (A)} + p_2 p_{11}^{\rm (II')} \pi_{4, W}^{\rm (II')}.
    \end{align*}
    Finally, when $p_1 + p_2 > 1$, the concordance function is
    \begin{align*}
        Q(J, W) &= p_{11} (1 - p_2) (2 \pi_{3, W} - 1) + p_{11} (2 \pi_{4, W} - 1) \\
        &\qquad + p_{11} (p_1 + p_2 - 1) - p_{01} (1 - p_1) - p_{10} (1 - p_2).
    \end{align*}

\end{theorem}

\begin{proof}
	The complete proof is reported in the~\hyperref[app]{Appendix}. 
\end{proof}

New formulations and estimators for Gini's gamma and Spearman's footrule can be mostly derived from Thereoms~\ref{thm:conc_upper} and \ref{thm:conc_lower}, though it remains necessary to obtain the values of $Q(M, \Pi)$ and $Q(W, \Pi)$. As we will observe in Section~\ref{sec:bounds}, these coincide with the bounds of Spearman's rho. It can be noted that when excess zeros are absent, meaning $p_1=p_2=0$, the expressions $Q(J, M)$ and $Q(J, W)$ simplify to their standard versions for continuous random variables. 

The concordance measures formulated in the previous theorems can be easily estimated by replacing $\{p_{00}, p_{10}, p_{01}, p_{11}\}$ and their analogous variants involving the corresponding regions with their relative frequencies. For the remaining entries, one could use the approximation below.
\begin{align*}
    C_M &\approx \frac{1}{|\mathcal C_{11}^{\rm (II)}|} \sum_{(x_i, y_i) \in \mathcal C_{11}^{\rm (II)}} P(x_i > X_2, y_i > Y_2\ |\ X_2 > 0, Y_2 > 0) \\
    &= \frac{1}{|\mathcal C_{11}^{\rm (II)}|} \sum_{(x_i, y_i) \in \mathcal C_{11}^{\rm (II)}} \min\{\tilde F_{\rm (II)}(x_i), \tilde G(y_i)\},
\end{align*}
and
\begin{align*}
    \pi_{4, M}^{\rm (I)} &\approx \frac{1}{|\mathcal C_{11}^{\rm (I)}|} \sum_{(x_i, y_i) \in \mathcal C_{11}^{\rm (I)}} P(x_i > X_2\ |\ J_{1111}^{\rm (I)}) \\
    &= \frac{1}{| \mathcal C_{11}^{\rm (I)} |} \sum_{(x_i, y_i) \in C_{11}^{\rm (I)}} \tilde F_{\rm (I)}(x_i), 
\end{align*}
where we define the sets $\mathcal C_{11}^{\rm (I)}$ and $\mathcal C_{11}^{\rm (II)}$ of observations where both $x$ and $y$ are strictly positive and $x$ satisfies $\rm (I)$ and $\rm (II)$ respectively. Furthermore, we have the distribution functions $\tilde F_{\rm (I)}$ and $\tilde F_{\rm (II)}$ of the variable $X$ given that $X$ satisfies $\rm (I)$ and $\rm (II)$ respectively, which, along with $\tilde G$, can be estimated using their respective empirical distributions. The remaining terms can be estimated in an analogous manner.

In the next section, we follow the same arguments presented in Theorem~\ref{thm:conc_upper} and \ref{thm:conc_lower} to derive an expression of Spearman's $\rho$.

\subsection{Correction to Spearman's rho for zero-inflated data}\label{sec:rho}

 As noted in the introduction, the expression for Spearman's rho presented in \cite{Pimentel2009kendall} contains inaccuracies. In the following theorem, we use the results of the previous section to correct it.

\begin{theorem}\label{thm:rho}
    Spearman's rho for $(X, Y)$ is
    \begin{align}\label{eq:spm_rho}
        \rho &= 3 p_{11}^3 Q(J, \Pi)_{11} + 3 p_{11}^2 p_{10} Q(J, \Pi)_{10} + 3 p_{11}^2 p_{01} Q(J, \Pi)_{01} + 3 p_{11} p_{10} p_{01} Q(J, \Pi)_{00} \notag\\
        &\quad\quad + 3 (p_{11} p_{00} - p_{10} p_{01}) + 3 p_{11} \left\{ p_{10} \left[2 \pi_{4, J} - 1\right] + p_{01} \left[2 \pi_{3, J} - 1\right] \right\},
    \end{align}
    where $Q(J, \Pi)_{ab}$ is the difference between the conditional probability of concordance and discordance given $S_{ab} = \{X_1 > 0, Y_1 > 0, X_2 > 0, Y_3 > 0, {\rm sign}\ Y_2 = a, {\rm sign}\ X_3 = b \}$ for $a, b \in \{0, 1\}$.
\end{theorem}

\begin{proof}
	The full proof is reported to the~\hyperref[app]{Appendix}. 
\end{proof}	

The formulation of Spearman's $\rho$ reported in Theorem~\ref{thm:rho} varies significantly from what was proposed by \cite{Pimentel2009kendall}. In fact, \cite{Pimentel2009kendall} suggested that
\begin{equation*}
    \rho = 3 p_{11} (1 - p_1)(1 - p_2) Q(J, \Pi)_P + 3 (p_{11} p_{00} - p_{10} p_{01}).
\end{equation*}
In this context, $Q(J, \Pi)_P$ represents the difference between the conditional probabilities of concordance and discordance when $\{X_1 > 0, Y_1 > 0, X_2 > 0, Y_3 > 0\}$ is satisfied, and it is considered to indicate Spearman's rho away from zero. 
This interpretation is mistaken because $X_3$ and $Y_2$ might still be zero. This fact is clear in the derived Eq.~\eqref{eq:spm_rho}, where it was important to differentiate among the cases $Q(J, \Pi)_{11}$, $Q(J, \Pi)_{10}$, $Q(J, \Pi)_{01}$, and $Q(J, \Pi)_{00}$ to accurately formulate $\rho$.

An estimator of $\rho$ can be obtained by replacing the probabilities with their respective relative frequencies. Moreover, $3 Q(J, \Pi)_{11}$, that is, Spearman's rho away from zero, can be estimated using the standard rank-based estimator on all observations where both $X$ and $Y$ are strictly positive. We notice that there exist other estimators of Spearman's $\rho$ that can handle both continuous, discrete, and mixed cases \citep{Remillard2024}. However, our proposed estimator performs well as highlighted in Section~\ref{sec:simulation} and Section~\ref{sec:application}.

As highlighted in the introduction, the presence of ties in the original data makes the estimation of the corresponding level of association insufficient to draw a conclusion about the strength of association. 
This is due to the attainable range being smaller than the usual range $[-1,1]$. Therefore, to assess the strength of association, it is also necessary to determine the attainable bounds of the corresponding measures, which we establish in the following section. 


\subsection{Best-possible bounds}\label{sec:bounds}

 Based on~\cite{Mesfioui2005properties}, as $\gamma$, $\phi$, and $\rho$ can be characterized using concordance functions, therefore the optimal upper and lower limits of these measures are attained at the Fr\'{e}chet-Hoeffding bounds of $C$. Utilizing this property, we establish attainable bounds for $\gamma$ and $\phi$, detailed in the following theorem.
	
	\begin{theorem}\label{thm:boundsgamma}
		The best-possible upper bound $\gamma_{\max}$ and lower bound $\gamma_{\min}$ for $\gamma$ are
        \begin{align*}
            \gamma_{\max} &= \begin{cases}
                \frac{2}{3} + \frac{1}{3} \max\{p_1, p_2\} \left(3 - 6 \max\{p_1, p_2\} + \max\{p_1, p_2\}^2\right), & p_1 + p_2 \geq 1, \\
                1 - \max\{p_1, p_2\}^2 - \frac{1}{3} \min\{p_1, p_2\}^3, & p_1 + p_2 < 1,
            \end{cases} \\
            \gamma_{\min} &= \begin{cases}
                (p_1 + p_2 + \min\{p_1, p_2\} - 2) (1 - \max\{p_1, p_2\}) - \frac{1}{3} (1 - \max\{p_1, p_2\}^3), & p_1 + p_2 \geq 1, \\
                (1 - p_1 - p_2)^2 - 2 (1 - p_1)(1 - p_2) - \frac{1}{3} \min\{p_1, p_2\}^3, & p_1 + p_2 < 1.
            \end{cases}
        \end{align*}
The best-possible upper bound $\phi_{\max}$ and lower bound $\phi_{\min}$ for $\phi$ are
\begin{align*}
    \phi_{\max} &= 1 - \frac{3}{2} \max\{p_1, p_2\}^2 + \frac{1}{2} \max\{p_1, p_2\}^3, \\
    \phi_{\min} &= \frac{3}{2} \max\{0, p_1 + p_2 - 1\} (1 - \max\{p_1, p_2\}) - \frac{1}{2} \left(1 - \max\{p_1, p_2\}^3\right).
\end{align*}	
\end{theorem}
	
\begin{proof}
    The proof is presented in the~\hyperref[app]{Appendix}.
\end{proof}

Several methods exist to derive the upper and lower bounds on $\gamma$ and $\phi$. The simplest involves specifying the joint distribution $J$ of $(X, Y)$ within the framework of Theorems~\ref{thm:conc_upper} and \ref{thm:conc_lower} and calculating each term accordingly. An in-depth proof using this approach is provided in the ~\hyperref[app]{Appendix}.
Another approach is to utilize the findings from \cite{Mesfioui2005properties}, which indicate that a concordance function can be represented as
\begin{equation}\label{eq:conc_general}
    Q(J, J_\star) = \mathbb E_J[J_\star(X, Y)] + \mathbb E_J[J_\star (X, Y^-)] + \mathbb E_J[J_\star(X^-, Y)] + \mathbb E_J[J_\star(X^-, Y^-)] - 1.
\end{equation}
This approach enables the derivation of the bounds through the distribution functions of $X$ and $Y$, as outlined in Eq.~\eqref{eq:distr}. Additional details can be found in the supplementary file.

\begin{figure}[p]
\vspace*{-1cm}
	\begin{subfigure}{.52\textwidth}
		\centering
		\includegraphics[width=2.8in]{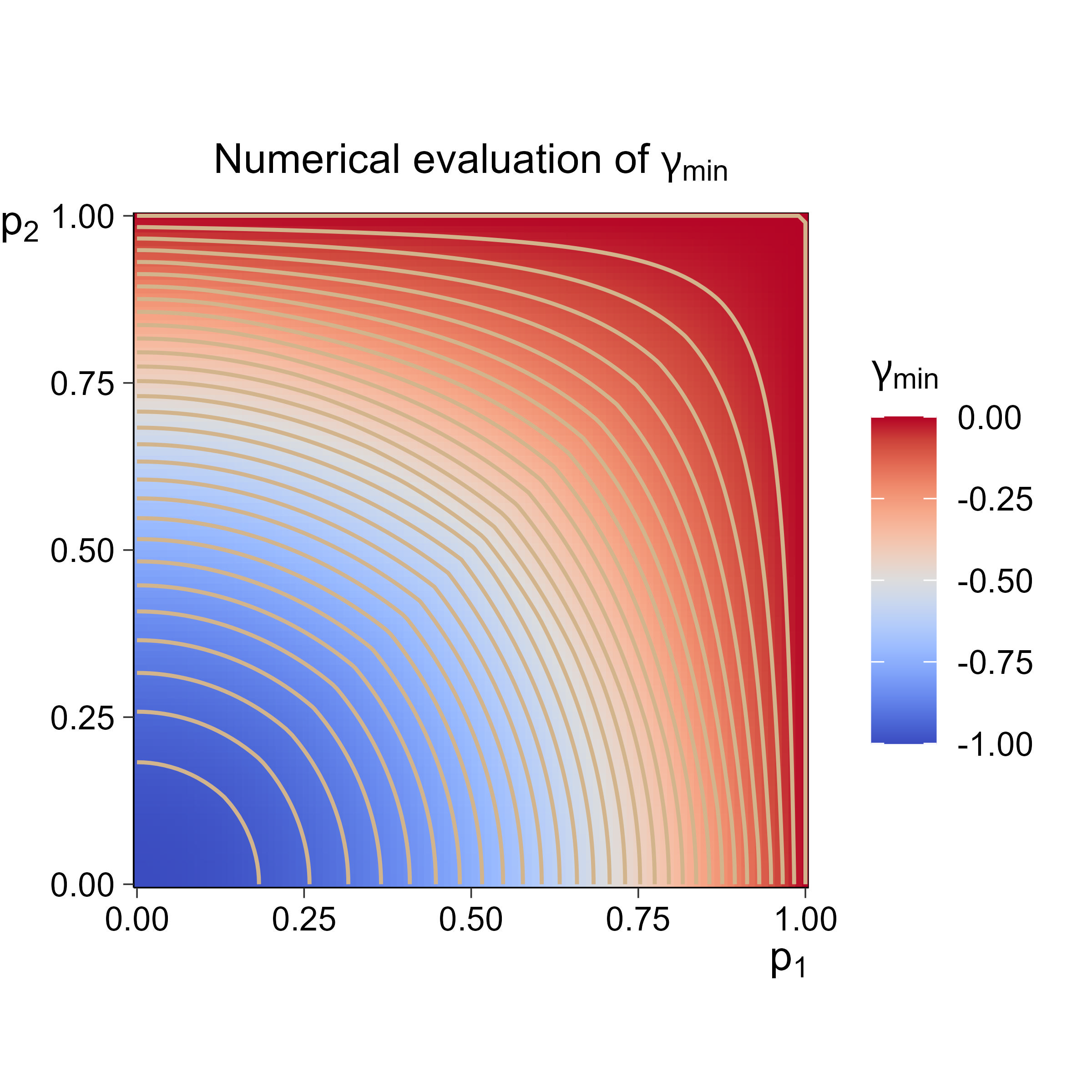}
		\label{fig:gammalower}
		\vspace{0cm}
	\end{subfigure}%
	\begin{subfigure}{.52\textwidth}
		\centering
		\includegraphics[width=2.8in]{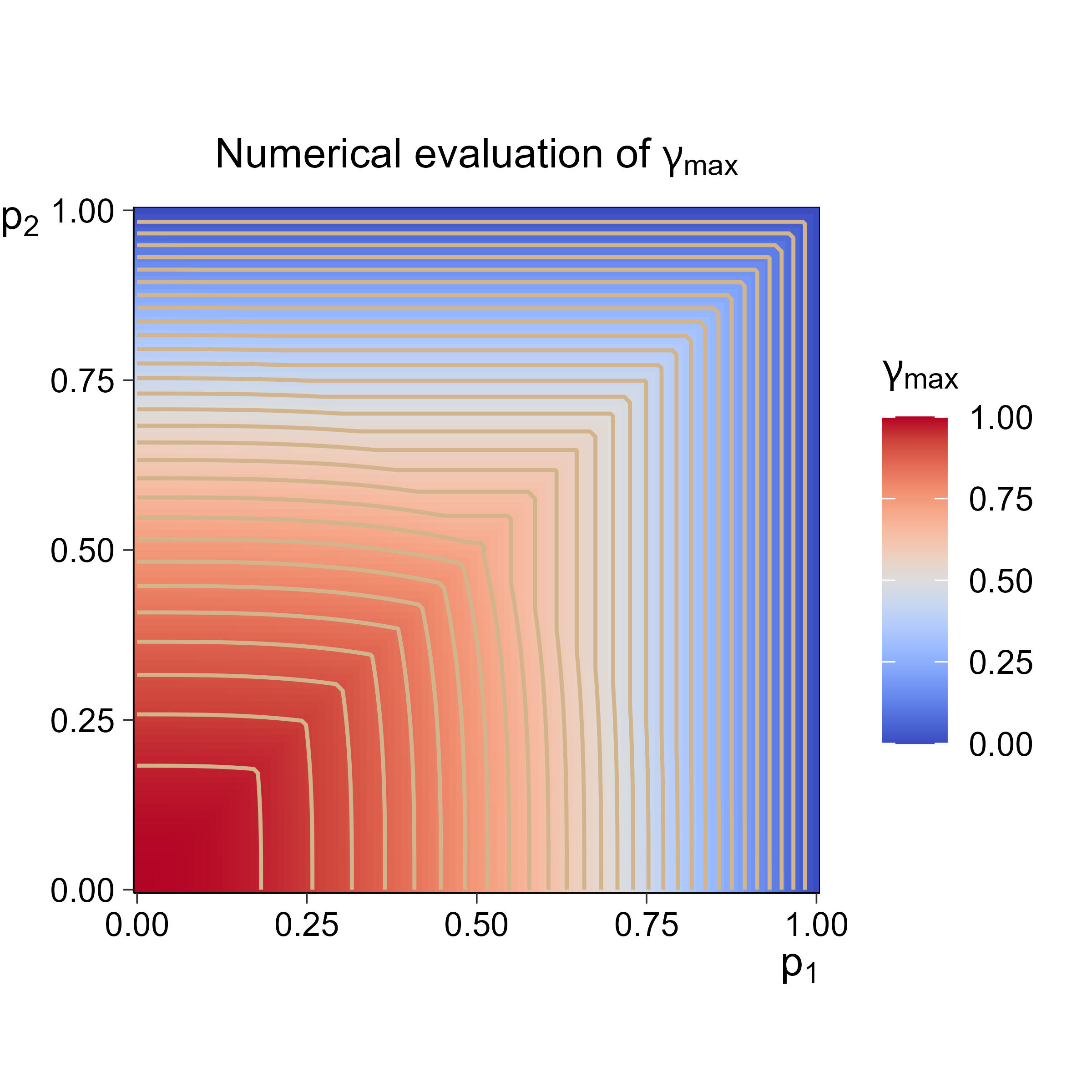}
		\label{fig:gammaupper}
		\vspace{0cm}
	\end{subfigure}
	\vspace{-1.5cm}
	\caption{Numerical evaluation of the lower and upper bounds for Gini's gamma.} 
	\label{fig:gamma}
\end{figure}

\begin{figure}[p]
	\begin{subfigure}{.52\textwidth}
		\centering
		\includegraphics[width=2.8in]{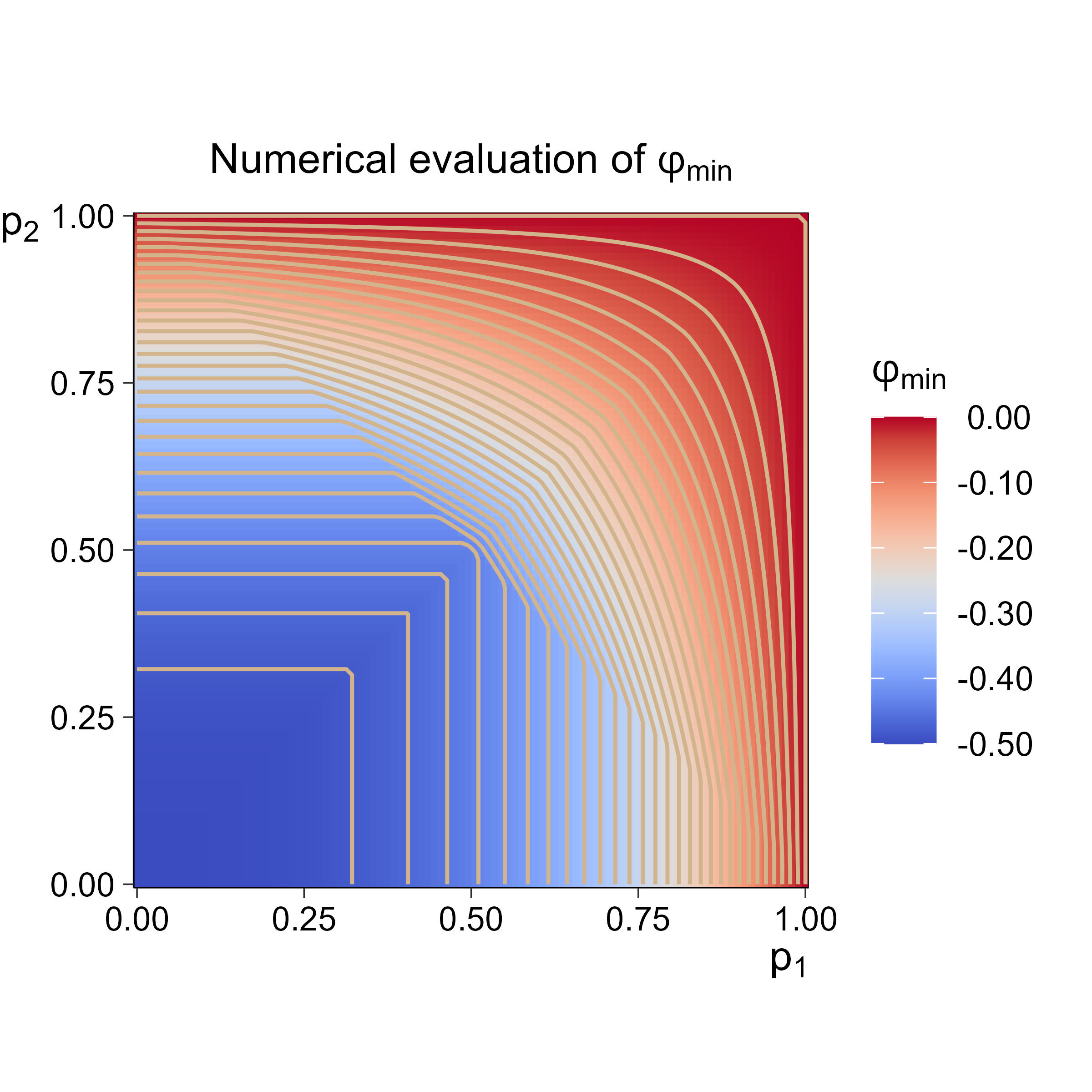}
		\label{fig:philower}
		\vspace{0cm}
	\end{subfigure}%
	\begin{subfigure}{.52\textwidth}
		\centering
		\includegraphics[width=2.8in]{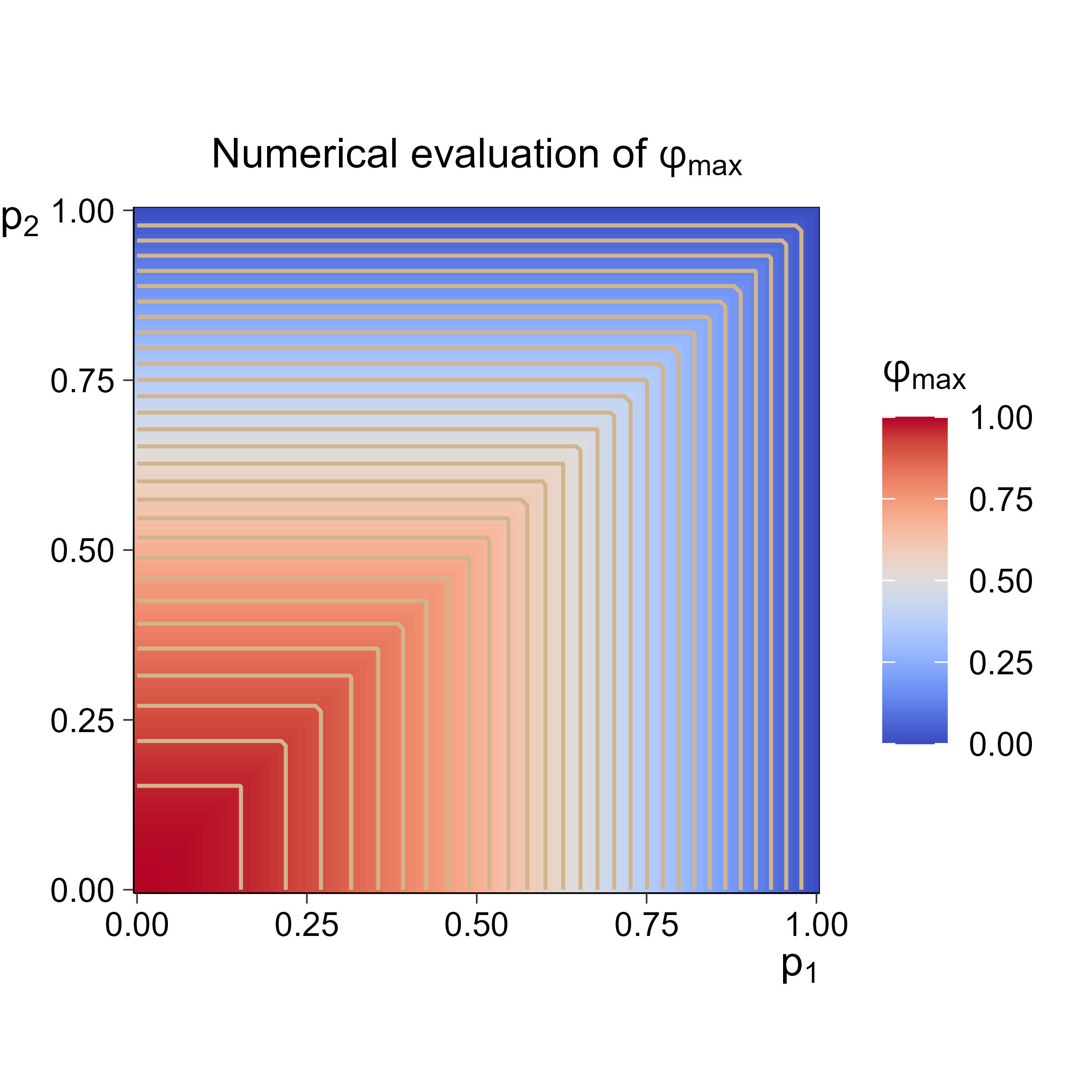}
		\label{fig:phiupper}
		\vspace{0cm}
	\end{subfigure}
	\vspace{-1.5cm}
	\caption{Numerical evaluation of the lower and upper bounds for Spearman's footrule.} 
	\label{fig:phi}
\end{figure}

\begin{figure}[p]
	\begin{subfigure}{.52\textwidth}
		\centering
		\includegraphics[width=2.8in]{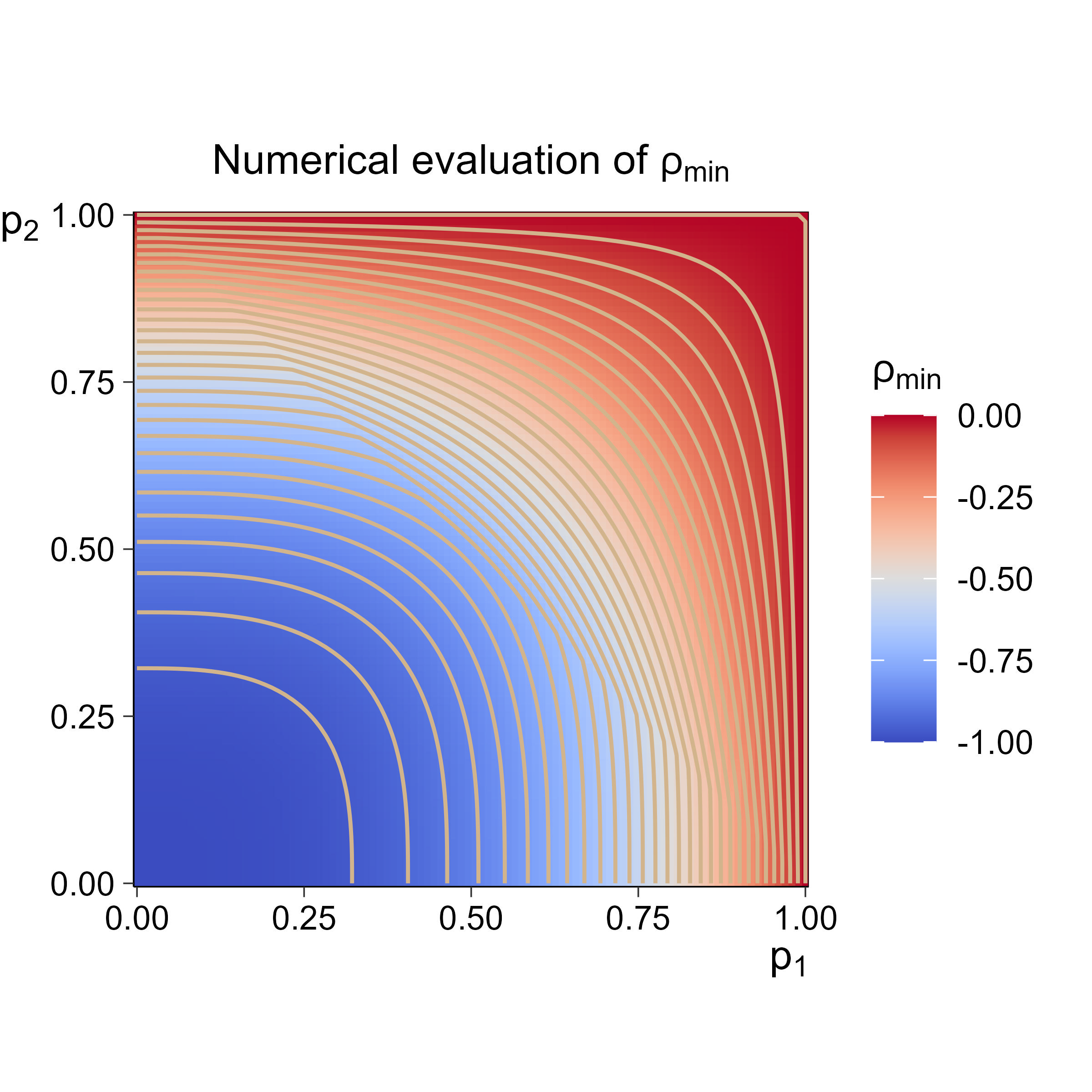}
		\label{fig:rholower}
		\vspace{0cm}
	\end{subfigure}%
	\begin{subfigure}{.52\textwidth}
		\centering
		\includegraphics[width=2.8in]{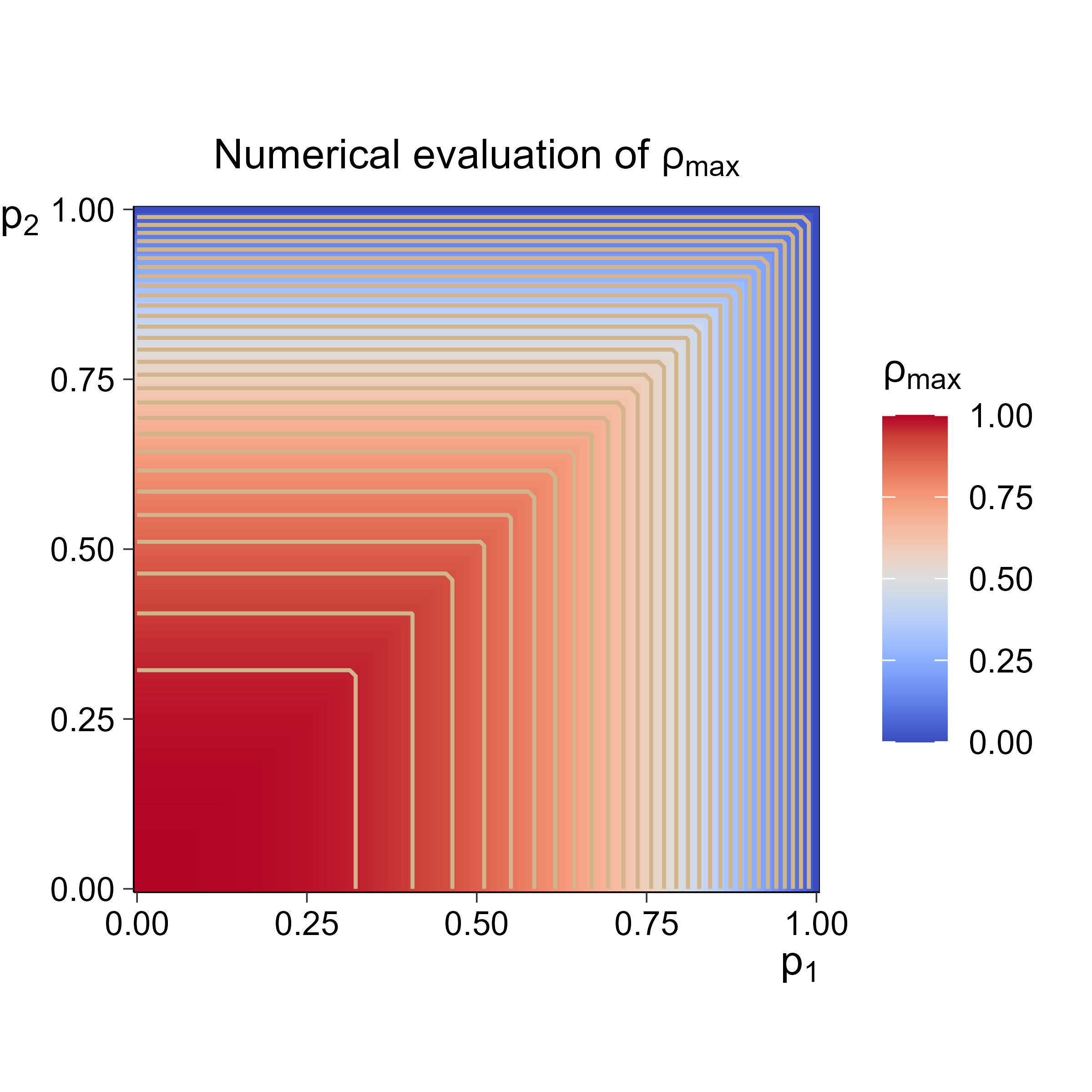}
		\label{fig:rhoupper}
		\vspace{0cm}
	\end{subfigure}
	\vspace{-1.5cm}
	\caption{Numerical evaluation of the lower and upper bounds for Spearman's rho.} 
	\label{fig:rho}
\end{figure}

As outlined earlier, \cite{Mesfioui2022rho} calculated the bounds of Spearman's $\rho$ using the formula provided by \cite{Pimentel2009kendall}. Due to inaccuracies within that formula, the boundaries given in \cite{Mesfioui2022rho} are likewise flawed. Here, we establish the accurate bounds for Spearman's $\rho$ through the expression detailed in Theorem~\ref{thm:rho}.

\begin{theorem}\label{thm:sper-rho-bounds}
	The best possible upper and lower bounds of Spearman's $\rho$ for $(X,Y)$ are
\begin{align*}
		\rho_{\max}&= 1- \max\{p_1^3, p_2^3\},\\
	\rho_{\min}&=\begin{cases}
	p_1^3+p_2^3-1,\quad & p_1+p_2 \le 1,\\
	-3(1-p_1)(1-p_2), \quad & p_1+p_2> 1.
\end{cases}
\end{align*}	
\end{theorem}
\begin{proof}
    The proof is delegated to the \hyperref[app]{Appendix}
\end{proof}

Similarly to Theorem~\ref{thm:boundsgamma}, there are multiple ways to prove this result. 
Comprehensive derivations based on the equation from Theorem~\ref{thm:rho} are available in \hyperref[app]{Appendix}.
Furthermore, replacing the margins in Eq.~\eqref{eq:distr} within the general expression of the concordance function presented in Eq.~\eqref{eq:conc_general} produces the same results. 
\cite{Mesfioui2022rhodisc} recently developed a general formula for determining the precise upper and lower bounds of Spearman's rho in the context of discrete variables. The bounds discussed in this paper can also be derived by integrating their structure with the margins of $X$ and $Y$ outlined in Eq.~\eqref{eq:distr}. A more comprehensive explanation of these two methods is provided in the supplementary file.

Figures~\ref{fig:gamma},~\ref{fig:phi}, and~\ref{fig:rho} depict the effects caused by zero-inflation on the bounds of Gini's gamma, Spearman's footrule, and Spearman's rho, respectively, using numerical visualization. Upon examining these figures, we notice that
\begin{enumerate}
\item when $p_1$ or $p_2$ equals 1, \ja{either one of the random variables $X$ and $Y$ is zero with probability one}, resulting in $\gamma = \phi = \rho =0$;

\item when $p_1 = p_2 = 0$, the attainable bounds are $\gamma_{\min}= \rho_{\min} = -1$, $\phi_{\min}=-1/2$, and $\gamma_{\max}=\phi_{\max} = \rho_{\max} = 1$. \ja{This is consistent with the known bounds for continuous random variables.}
\end{enumerate}
It is important to note that the bounds of all three concordance measures follow a similar pattern.


\section{Simulation study}
\label{sec:simulation}

We now conduct a simulation study to evaluate the effectiveness of the proposed theory in controlled settings. An implementation of this study in \textbf{R}~\citep{Rsoftware} is available on GitHub at \url{https://github.com/JasperArends/ConcZID}.
Here, we investigate the performance of the results in Section~\ref{sec:results} in a setting similar to \cite{perrone2023}. 
In particular, we compute the values of the estimators of $\gamma$, $\phi$ and $\rho$ for $N$ pairs generated from two random variables joined through the Fr\'echet copula $C_\alpha = (1 - \alpha) u v + \alpha \min\{u, v\}$, where $u, v, \alpha \in [0, 1]$ \citep{Nelsen2006}. 
We compare the estimates derived with the true values of $\gamma$, $\phi$ and $\rho$ calculated using the closed formula $$Q(C_\alpha, J_\star) = (1 - \alpha) Q(\Pi, J_\star) + \alpha Q(M, J_\star),$$ where $J_\star$ coincides with one of the Fr\'echet-Hoeffding copula bounds or the independence copula. Within this setting, a total of three parameters need to be chosen. We consider various levels of zero inflation by setting $p_1, p_2$ in $\{0.2, 0.8\}$.
The copula parameter $\alpha$ primarily governs the level of association, with higher values corresponding to stronger association. In the study, we consider three different values for $\alpha$, namely $\alpha \in \{0.2, 0.5, 0.8\}$ to depict different strength of association. 

\begin{figure}[t]
    \centering
    \includegraphics[width=.3\textwidth]{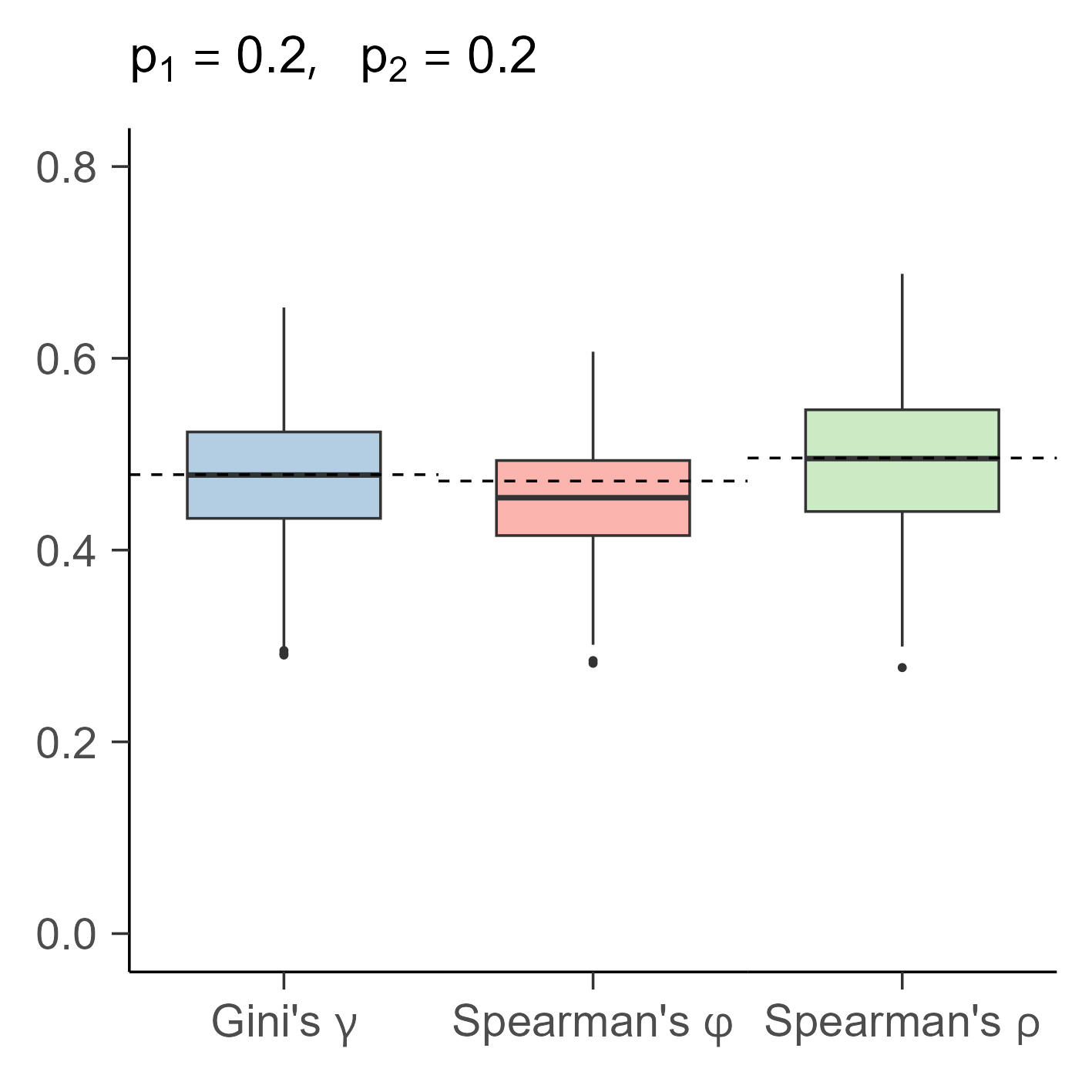}
    \includegraphics[width=.3\textwidth]{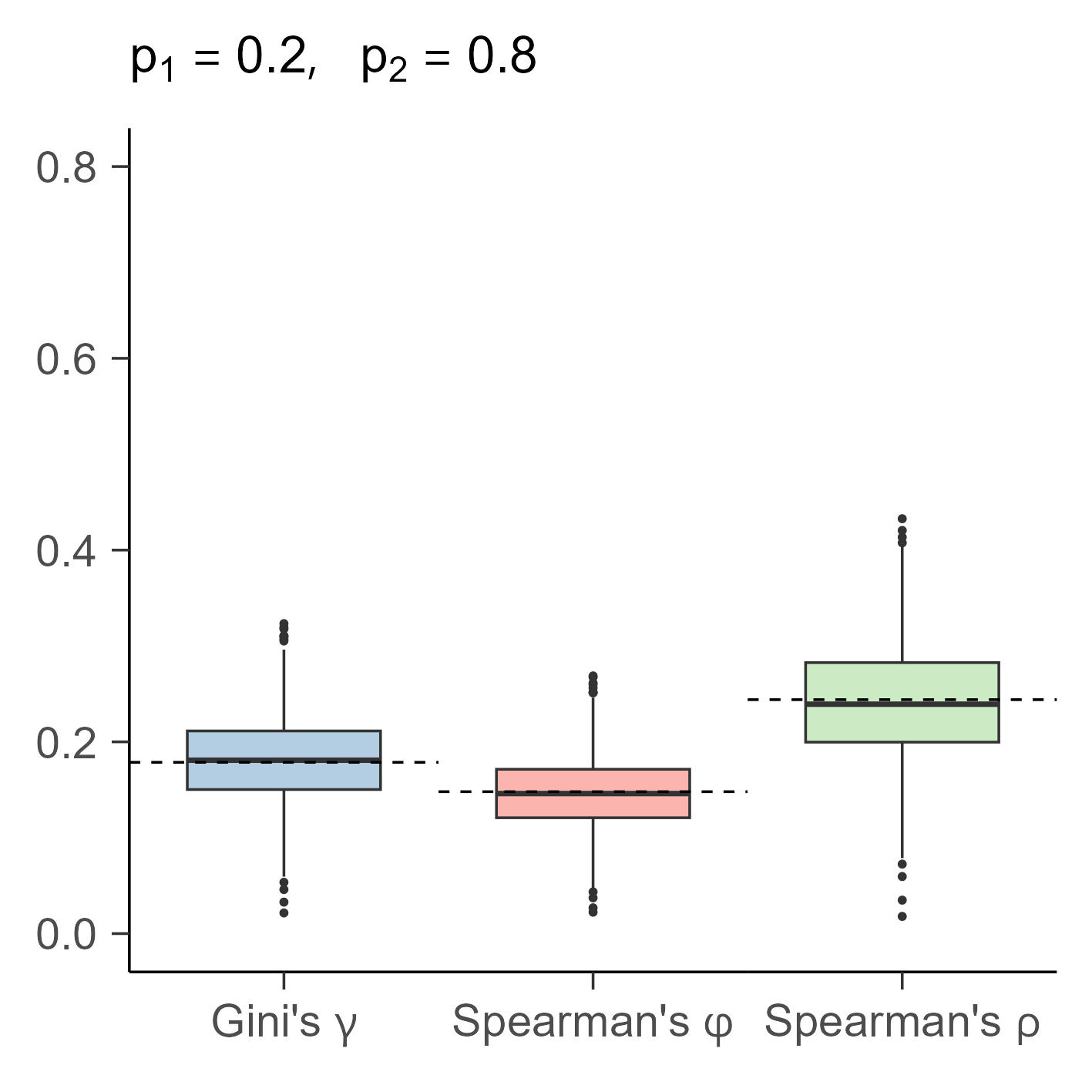}
    \includegraphics[width=.3\textwidth]{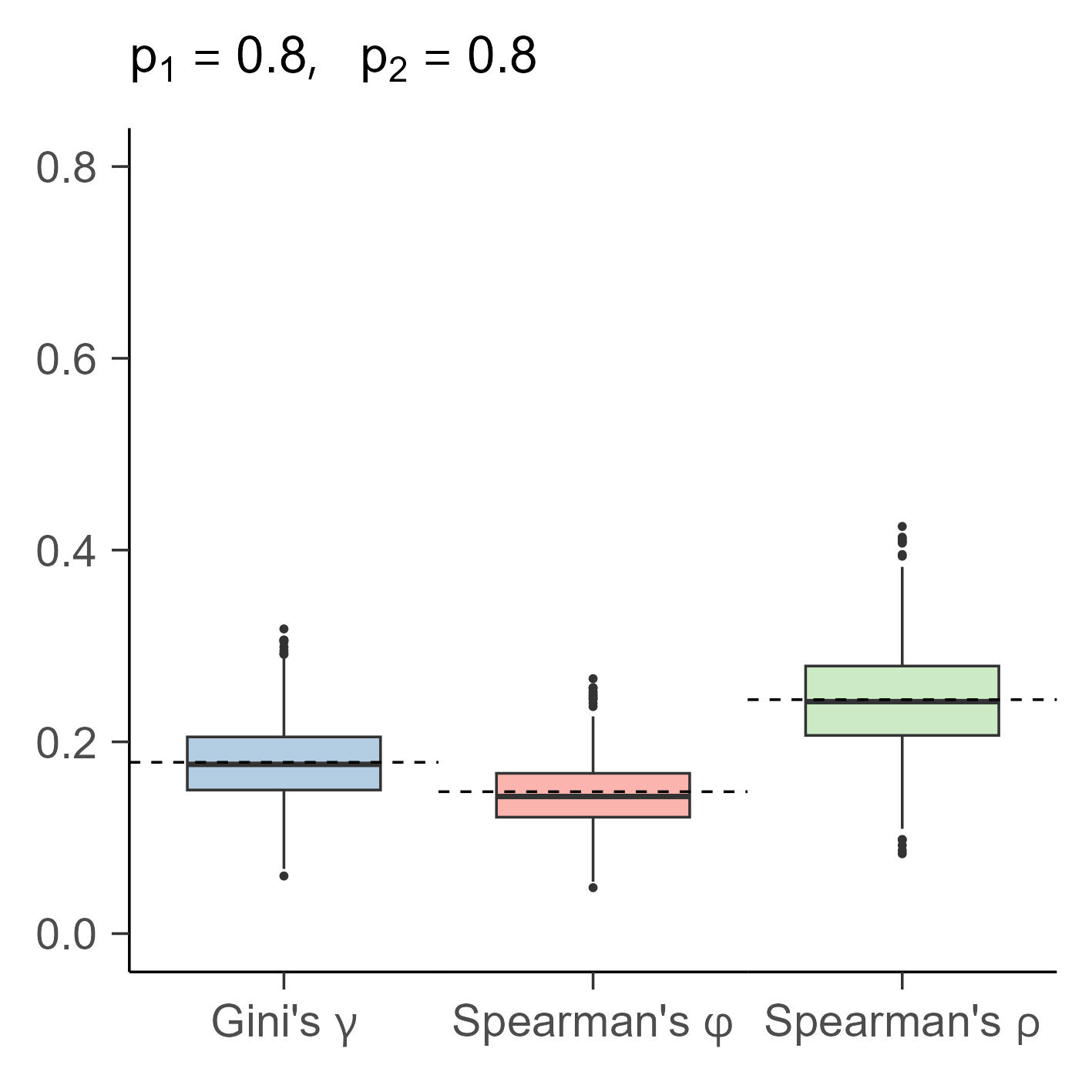}
    \caption{Simulation results for the Fr\'echet copula with $\alpha = 0.5$ and $N = 150$, averaged over 1000 runs. The dashed lines mark the true values of the concordance measures.}
    \label{fig:boxplots}
\end{figure}

\begin{table}[t!]
\centering
\caption{Simulation results for a sample size $N = 150$ and 1000 repetitions. The reported MSE is multiplied by a factor $10^2$.}
\begin{tabular}{lllllllllllllll}
\toprule[1.5pt]
$p_1$ & $p_2$ & $\alpha$ &  & \multicolumn{3}{l}{Gini's gamma}                   &  & \multicolumn{3}{l}{Spearman's   footrule}          &  & \multicolumn{3}{l}{Spearman's rho}               \\ \cline{5-7} \cline{9-11} \cline{13-15} 
      &       &          &  & True $\gamma$  & \multicolumn{2}{l}{$\hat \gamma$} &  & True $\phi$      & \multicolumn{2}{l}{$\hat \phi$} &  & True $\rho$    & \multicolumn{2}{l}{$\hat \rho$} \\ \cline{6-7} \cline{10-11} \cline{14-15} 
      &       &          &  &                & Mean            & MSE*            &  &                  & Mean            & MSE*          &  &                & Mean           & MSE*           \\ \hline
0.2   & 0.2   & 0.2      &  & 0.191          & 0.196           & 0.466           &  & 0.189            & 0.182           & 0.314         &  & 0.198          & 0.189          & 0.704          \\
      &       & 0.5      &  & 0.479          & 0.478           & 0.399           &  & 0.472            & 0.454           & 0.347         &  & 0.496          & 0.493          & 0.557          \\
      &       & 0.8      &  & 0.766          & 0.765           & 0.241           &  & 0.755            & 0.738           & 0.251         &  & 0.794          & 0.791          & 0.300          \\ \hline
0.2   & 0.8   & 0.2      &  & 0.071          & 0.074           & 0.172           &  & 0.059            & 0.059           & 0.106         &  & 0.098          & 0.094          & 0.348          \\
      &       & 0.5      &  & 0.179          & 0.181           & 0.226           &  & 0.148            & 0.147           & 0.154         &  & 0.244          & 0.241          & 0.408          \\
      &       & 0.8      &  & 0.286          & 0.288           & 0.241           &  & 0.237            & 0.236           & 0.180         &  & 0.390          & 0.385          & 0.384          \\ \hline
0.8   & 0.8   & 0.2      &  & 0.071          & 0.072           & 0.109           &  & 0.059            & 0.058           & 0.070         &  & 0.098          & 0.098          & 0.212          \\
      &       & 0.5      &  & 0.179          & 0.178           & 0.174           &  & 0.148            & 0.145           & 0.126         &  & 0.244          & 0.244          & 0.308          \\
      &       & 0.8      &  & 0.286          & 0.285           & 0.241           &  & 0.237            & 0.233           & 0.191         &  & 0.390          & 0.388          & 0.387          \\ \bottomrule[1.5pt]
\end{tabular}
\label{tab:simulations}
\end{table}

In Table~\ref{tab:simulations}, we present both the mean and the mean square error (MSE) for each of the estimators for a sample size of $N=150$ and $1000$ repetitions. Additionally, Figure~\ref{fig:boxplots} illustrates these outcomes through boxplots.
Examining Table~\ref{tab:simulations}, we notice that our estimators consistently perform well across all scenarios. 
It is particularly significant that their accuracy improves notably in scenarios of substantial zero-inflation, like when $p_1 = p_2 = 0.8$. 
The improvement results largely from the fact that the estimators are specifically designed to handle zero-inflation, demonstrating the method's efficacy for the intended scenarios.

In addition to assessing how well the estimators for $\gamma$, $\phi$, and $\rho$ perform, we investigate the estimation of their respective bounds. We focus on a similar situation by forming a sample of $N$ pairs originating from two random variables where each variable is zero-inflated with probabilities $p_1$ and $p_2$, where $p_1, p_2 \in \{0.2, 0.8\}$.
Table~\ref{tab:sim_bounds} presents the average results calculated from 1000 iterations. The estimated bounds match the theoretical bounds in almost all cases, confirming that our proposed method works well.

\begin{table}[t!]
\caption{Simulation results of the bounds estimation for $N = 150$ samples, averaged over 1000 runs.}
\begin{tabular}{lllllllllllll} \toprule[1.5pt]
$p_1$ & $p_2$ &  & \multicolumn{3}{l}{Gini's gamma}                             &  & \multicolumn{3}{l}{Spearman's   footrule}                       &  & \multicolumn{2}{l}{Spearman's rho}     \\ \cline{4-6} \cline{8-10} \cline{12-13} 
      &       &  & True                 & \multicolumn{2}{l}{Estimated}         &  & True                    & \multicolumn{2}{l}{Estimated}         &  & True               & Estimated         \\ \hline
0.2   & 0.2   &  & $[-0.92, 0.96]$    & \multicolumn{2}{l}{$[-0.92, 0.96]$} &  & $[-0.50, 0.94]$       & \multicolumn{2}{l}{$[-0.50, 0.94]$} &  & $[-0.98, 0.99]$, & $[-0.98, 0.99]$ \\
0.2   & 0.8   &  & $[-0.32, 0.36]$    & \multicolumn{2}{l}{$[-0.32, 0.36]$} &  & $[-0.25, 0.30]$       & \multicolumn{2}{l}{$[-0.24, 0.30]$} &  & $[-0.48, 0.49]$  & $[-0.47, 0.49]$ \\
0.8   & 0.8   &  & $[-0.08, 0.36]$    & \multicolumn{2}{l}{$[-0.08, 0.36]$} &  & $[-0.06, 0.30]$       & \multicolumn{2}{l}{$[-0.07, 0.30]$} &  & $[-0.12, 0.49]$  & $[-0.12, 0.49]$ \\ \bottomrule[1.5pt]
\end{tabular}
\label{tab:sim_bounds}
\end{table}

\section{Case study}\label{sec:application}

In this section, we illustrate how the theory proposed here can be used to analyze real data. 
To demonstrate the broad applicability of our theory, we consider two different datasets from weather forecasting and insurance, respectively. The first dataset groups the total precipitation recorded between 12:00 and 13:00 hours over the period from 2020 until 2023 at two weather stations in the Netherlands, that is, Schiphol and De Bilt. 
The data is freely available through the Royal Netherlands Meteorological Institute \citep{KNMI2024}, and it is depicted in Figure \ref{fig:KNMI}. 
The absence of rain on many days induces the inflation to zero, making the data set fit the scope of this work. 
Although the precipitation is zero-inflated continuous in nature, the measurements are rounded to 0.1 millimeters and are therefore technically discrete. We solve the issue of ties away from zero by ranking the tied observations at random. Since the inflation at zero forms a substantial amount of the data, the estimators will still yield proper approximations.
The estimated probabilities of no rain at Schiphol and De Bilt are $\hat p_1 = 0.884$ and $\hat p_2 = 0.868$, respectively, showing a heavy zero inflation as expected. 
The estimates of the concordance measures are given in Table~\ref{table:application}, and are small and positive. 
However, the upper bounds suggest that the association between precipitation at both weather stations appears much stronger than previously thought when interpreting the measure estimates on their own.

 The second dataset originates from insurance and consists in the Danish fire losses dataset, which was analyzed by~\cite{Denuit2017bounds} in the context of Kendall's tau for zero-inflated settings. 
 This dataset is available in the \textbf{R}~\citep{Rsoftware} package \texttt{fitdistrplus}~\citep{FitdistriplusPackage}. The data was collected by Copenhagen Reinsurance and consists of 2167 fire loss records from the period 1980 to 1990, adjusted for inflation.
In this study, we focus on content-related losses, where the proportion of zero values is $0.225$ for losses and $0.716$ for profits. 
These proportions indicate that a substantial portion of the observations are zero, reinforcing the need for zero-inflated continuous modeling techniques. 
Table~\ref{table:application} presents the estimates of concordance measures along with their attainable bounds.

By comparing the estimated measures with their attainable bounds, we observe that all the estimated values of $\rho$, $\gamma$ and $\phi$ lie well within their respective lower and upper limits. 
Like the precipitation dataset, the measures here reveal a positive association that seems considerably stronger within the narrow range of attainable values, thus highlighting the effectiveness of the methods suggested in this study.

\begin{figure}[t]
    \centering
    \includegraphics[width=.7\linewidth]{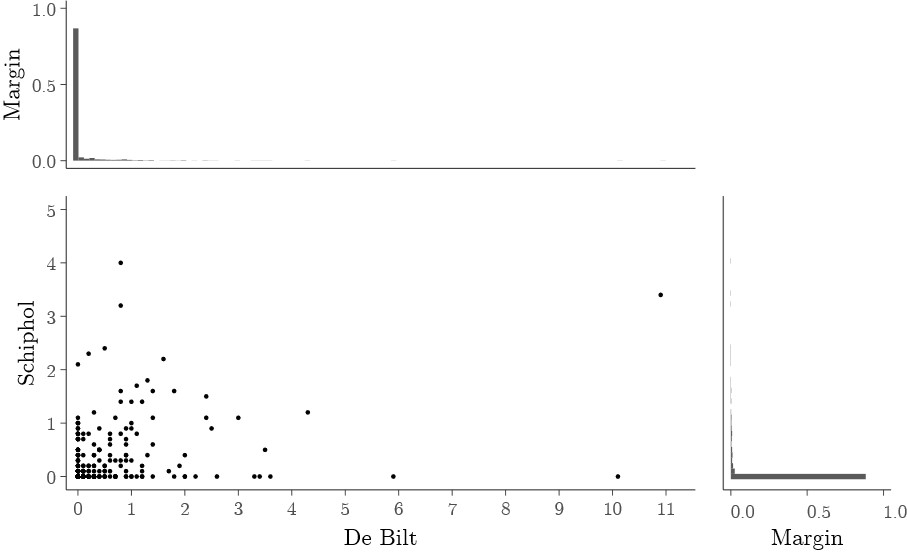}
    \caption{Scatter plot of historical records for precipitation (in mm) at De Bilt and Schiphol with a strong indication of zero inflation.}
    \label{fig:KNMI}
\end{figure}

\begin{table}[t]
	\centering
	\caption{Estimates of measures of association and their attainable bounds for the precipitation dataset.}
		\begin{tabular}{lcccccccc}
			\toprule[1.5pt]
            && \multicolumn{3}{l}{Precipitation dataset} && \multicolumn{3}{l}{Insurance dataset} \\ \cline{3-5} \cline{7-9}
			                     && $\rho$ & $\gamma$ & $\phi$ && $\rho$ & $\gamma$ & $\phi$  \\ \hline
			Estimate              && 0.183  & 0.126    & 0.098  && 0.273  & 0.183    &  0.153  \\
            Estimated lower bound && -0.047 & -0.031   & -0.024 && -0.622 & -0.441   &  -0.317 \\
            Estimated upper bound && 0.309  & 0.218    & 0.173  && 0.633  & 0.484    & 0.415   \\
			\bottomrule[1.5pt]
		\end{tabular}
	\label{table:application}
\end{table}

\section{{Conclusions}}\label{sec:conclusion}
In this paper, we present new formulations for Gini's gamma, Spearman's footrule, and Spearman's rho in the context of zero-inflated distributions. 
Based on these formulations, we propose novel estimators for these concordance measures. 
In addition, we derive the range of attainable values for the concordance measures and suggest a method to estimate them, making the proposed estimators interpretable and practically applicable. 
We evaluate how the estimators perform with both simulated and real-world data under different conditions. Our results confirm the effectiveness of the methods introduced in this work.

In future research, we aim to extend our approach to derive new formulations of Gini's gamma and Spearman's footrule for zero-inflated discrete settings. This direction builds on prior work for Kendall's tau by \cite{perrone2023}, who proposed a new formulation and estimator tailored to zero-inflated discrete distributions, along with corresponding bounds.

Another promising avenue, partially explored in \cite{Arends2025}, involves further investigating the relationship between Kendall's tau and Spearman's rho in zero-inflated contexts. This connection has been well established in the fully continuous case \cite{Schreyer2017}, but remains largely unexamined under zero-inflation. 
Understanding it in this setting can uncover deeper structural relationships among concordance measures, enabling more coherent interpretations and potentially streamlining estimation strategies across different measures.

Overall, our contributions lay the groundwork for a more robust understanding and estimation of concordance measures under zero-inflation. 
By addressing both theoretical and practical challenges, this line of research enhances the analysis of associations in data with zero inflation, which is a common feature in many applied fields, including economics, insurance, and weather forecasting.

\bibliographystyle{plain}
\bibliography{References_TML}

\section*{Appendix: Proofs of the results} \label{app}
\addcontentsline{toc}{section}{Appendix: Proofs} 
\renewcommand{\thesubsection}{A.\arabic{subsection}} 

\subsection{Proof of Theorems ~\ref{thm:conc_upper}, ~\ref{thm:conc_lower} and ~\ref{thm:rho}}
For Gini's gamma and Spearman's footrule, expressions for $Q(M, \Pi)$ and $Q(W, \Pi)$ follow from the upper and lower bounds on Spearman's rho respectively. 
Therefore, we are only interested in finding formulations for $Q(J, J_\star)$ where $J_\star$ where $J_\star$ amounts to the Fr\'echet-Hoeffding bounds or independence copula. 
We will first derive a general formulation for $Q(J, J_\star)$ without specifying $J_\star$ and then further investigate the distributions of $X, Y$ considering each of these joint distributions separately.

The proof is along the same lines as the method presented by \cite{Pimentel2015kendall}. The strategy is to consider four distinct scenarios of $(X, Y)$, i.e., $\{X = 0, Y = 0\}$, $\{X > 0, Y = 0\}$, $\{X = 0, Y > 0\}$ and $\{X > 0, Y > 0\}$. To that end, define
\begin{align*}
	p^\star_{00}&=P(X_2=0, Y_2=0), \quad p^\star_{01}=P(X_2=0, Y_2>0),  \notag\\
	p^\star_{10}&=P(X_2>0, Y_2=0), \quad p^\star_{11}=P(X_2>0, Y_2>0)
\end{align*}
and consider the set $J_{a_1b_1a_2b_2}=\{\rm{sign}\,X=a_1,\,\rm{sign}\,Y=b_1,\ \rm{sign}\,X_\star=a_2,\,\rm{sign}\, Y_\star=b_2\}$,
where $a_1, b_1, a_2, b_2 \in \{0, 1\}$. Then, using the law of total probability, one may find that
\begin{align*}
    P(X < X_\star, Y < Y_\star) &= P(X < X_\star, Y_1 < Y_\star\ |\ J_{1111}) p_{11} p_{11}^\star +  P(X < X_\star\ |\ J_{1011}) p_{10} p_{11}^\star \\
	&\quad +  P(Y < Y_\star\ |\ J_{0111}) p_{01} p_{11}^\star + p_{00} p_{11}^\star, \\
    P(X > X_\star, Y > Y_\star) &= P(X > X_\star, Y > Y_\star\ |\ J_{1111}) p_{11} p_{11}^\star +  P(Y > Y_\star\ |\ J_{1101}) p_{11} p_{01}^\star \\
	&\quad +  P(X > X_\star\ |\ J_{1110}) p_{11} p_{10}^\star + p_{11} p_{00}^\star, \\
	P(X < X_\star, Y > Y_\star) &= P(X < X_\star, Y > Y_\star\ |\ J_{1111}) p_{11} p_{11}^\star +  P(Y > Y_\star\ |\ J_{0111}) p_{01} p_{11}^\star \\
	&\quad + P(X < X_\star\ |\ J_{1110}) p_{11} p_{10}^\star + p_{01} p_{10}^\star, \\
	P(X > X_\star, Y < Y_\star) &= P(X_1 > X_\star, Y < Y_\star\ |\ J_{1111}) p_{11} p_{11}^\star +  P(Y < Y_\star\ |\ J_{1101}) p_{11} p_{01}^\star \\
	&\quad + P(X > X_\star\ |\ J_{1011}) p_{10} p_{11}^\star + p_{10} p_{01}^\star.
\end{align*}
We can substitute these into the definition of the concordance function and obtain
\begin{align*}
    Q(J, J_\star) &= P(X > X_\star, Y > Y_\star) + P(X < X_\star, Y < Y_\star) \\
    &\quad - P(X > X_\star, Y < Y_\star) - P(X < X_\star, Y > Y_\star) \\
    &= p_{11} p_{11}^\star \bigl\{P(X > X_\star, Y > Y_\star\ |\ J_{1111}) + P(X < X_\star, Y < Y_\star\ |\ J_{1111}) \\
    &\qquad\qquad - P(X > X_\star, Y < Y_\star\ |\ J_{1111}) - P(X < X_\star, Y > Y_\star\ |\ J_{1111})\bigr\} \\
    &\quad + p_{10} p_{11}^\star \left\{P(X < X_\star\ |\ J_{1011}) - P(X > X_\star\ |\ J_{1011})\right\} \\
    &\quad + p_{01} p_{11}^\star \left\{P(Y < Y_\star\ |\ J_{0111}) - P(Y > Y_\star\ |\ J_{0111})\right\} \\
    &\quad + p_{11} p_{10}^\star \left\{P(X > X_\star\ |\ J_{1110}) - P(X < X_\star\ |\ J_{1110})\right\} \\
    &\quad + p_{11} p_{01}^\star \left\{P(Y > Y_\star\ |\ J_{1101}) - P(Y < Y_\star\ |\ J_{1101})\right\} \\
    &\quad + p_{00} p_{11}^\star + p_{11} p_{00}^\star - p_{01} p_{10}^\star - p_{10} p_{01}^\star,
\end{align*}
from which it follows that
\begin{align} \label{eq:conc_zid_general}
    Q(J, J_\star) &= p_{11} p_{11}^\star Q(J, J_\star)_{11} \notag\\
    &\quad + p_{10} p_{11}^\star \{2 \pi_{5, J_\star} - 1\} + p_{01} p_{11}^\star \{2 \pi_{6, J_\star} - 1\} \notag\\
    &\quad + p_{11} p_{10}^\star \{2 \pi_{4, J_\star} - 1\} + p_{11} p_{01}^\star \{2 \pi_{3, J_\star} - 1\} \notag\\
    &\quad + p_{00} p_{11}^\star + p_{11} p_{00}^\star - p_{01} p_{10}^\star - p_{10} p_{01}^\star.
\end{align}
Here, we define
\begin{align*}
    \pi_{5, J_\star} &= P(\{X_\star\ |\ X_\star > 0, Y_\star > 0\} > \{X\ |\ X > 0, Y = 0\}), \\
    \pi_{6, J_\star} &= P(\{Y_\star\ |\ X_\star > 0, Y_\star > 0\} > \{Y\ |\ X = 0, Y > 0\}),
\end{align*}
and $Q(J, J_\star)_{11}$ is the difference between the conditional probabilities of concordance and discordance given that all random variables, i.e., $X$, $Y$, $X_\star$ and $Y_\star$ to be specific, are strictly positive. In the following, we will prove Theorem~\ref{thm:conc_upper}, ~\ref{thm:conc_lower} and \ref{thm:rho} by elaborating upon this expression and specifying the joint distribution $J_\star$.

\subsubsection{Proof of Theorem \ref{thm:conc_upper}}
We first consider the case where $(X_\star, Y_\star)$ is governed by the upper Fr\'echet-Hoeffding copula defined by $M(u, v) := \min\{u, v\}$ for all $u, v \in [0, 1]$. Without loss of generality, we assume that $p_1 \leq p_2$. As recalled in the main text, \cite{Denuit2017bounds} derived that $F(\{X_\star\,|\,X_\star > 0,\,Y_\star > 0\}) > p_2$ and $F(\{X_\star\,|\,X_\star > 0,\,Y_\star = 0\}) \leq p_2$. Therefore, we will elaborate upon the previous results in Eq.~\eqref{eq:conc_zid_general} by distinguishing cases where $X$ is either smaller or larger than $F^{-1}(p_2)$. To that end, let
\begin{align*}
    J_{a_1b_1a_2b_2}^{\rm (I)} = \{{\rm sign}\ X = a_1,  {\rm sign}\ Y = b_1, {\rm sign}\ X_\star = a_2, {\rm sign}\ Y_\star = b_2, F(X) \leq p_2\}, \\
    J_{a_1b_1a_2b_2}^{\rm (II)} = \{{\rm sign}\ X = a_1,  {\rm sign}\ Y = b_1, {\rm sign}\ X_\star = a_2, {\rm sign}\ Y_\star = b_2, F(X) > p_2\},
\end{align*}
for $a_1, b_1, a_2, b_2 \in \{0, 1\}$. Here, $\rm (I)$ indicates the region where $F(X) \leq p_2$ and $\rm (II)$ the part where $F(X) > p_2$. Additionally, consider the probabilities
\begin{align*}
	p^{\rm (I)}_{11} &= P(X > 0, F(X) \leq p_2, Y_1 > 0), \quad p^{\rm (II)}_{11} = P(F(X) > p_2, Y > 0),  \notag\\
	p^{\rm (I)}_{10} &= P(X > 0, F(X) \leq p_2, Y = 0), \quad p^{\rm (II)}_{10}=P(F(X) > 0, Y_\star = 0).
\end{align*}
We will first elaborate upon each of the terms in Eq.~\eqref{eq:conc_zid_general} individually by
\begin{align*}
    p_{11} p_{11}^\star Q(J, M)_{11} &= p_{11} p_{11}^\star \bigl\{P(X > X_\star, Y > Y_\star\ |\ J_{1111}) + P(X < X_\star, Y < Y_\star\ |\ J_{1111}) \\
    &\qquad\qquad - P(X > X_\star, Y < Y_\star\ |\ J_{1111}) - P(X < X_\star, Y > Y_\star\ |\ J_{1111})\bigr\} \\
    &= p_{11} p_{11}^\star \bigl\{4 P(X > X_\star, Y > Y_\star\ |\ J_{1111}) - 1 \\
    &\qquad\qquad + P(X < X_\star\ |\ J_{1111}) - P(X > X_\star\ |\ J_{1111}) \\
    &\qquad\qquad + P(Y < Y_\star\ |\ J_{1111}) -  P(Y > Y_\star\ |\ J_{1111})\bigr\} \\
    &= p_{11}^\star \left\{p_{11}^{\rm (II)} 4 P(X > X_\star, Y > Y_\star\ |\ J_{1111}^{\rm (II)}) - p_{11}\right\} \\
    &\quad + p_{11}^\star \bigl\{p_{11}^{\rm (I)} + p_{11}^{\rm (II)} \left[P(X < X_\star\ |\ J_{1111}^{\rm (II)}) - P(X > X_\star\ |\ J_{1111}^{\rm (II)})\right] \bigr\} \\
    &\quad + p_{11} p_{11}^\star \bigl\{P(Y < Y_\star\ |\ J_{1111}) -  P(Y > Y_\star\ |\ J_{1111})\bigr\}, \\
\end{align*}
and
\begin{align*}
    p_{10} p_{11}^\star \{2 \pi_{2, M} - 1\} &= p_{11}^\star \Bigl\{p_{10}^{\rm (II)} \left[ P\bigl(X < X_\star\ |\ J_{1011}^{\rm (II)}\bigr) - P\bigl(X > X_\star\ |\ J_{1011}^{\rm (II)}\bigr)\right] + p_{10}^{\rm (I)}\Bigr\}, \\
    p_{11} p_{10}^\star \{2 \pi_{4, M} - 1\} &= p_{10}^\star \Bigl\{p_{11}^{\rm (I)} \left[2 \pi_{4, M}^{\rm (I)} - 1\right] + p_{11}^{\rm (II)}\Bigr\}.
\end{align*}
We can substitute these results back to find that
\begin{align*}
    Q(J, M) &= p_{11}^\star \bigl\{p_{11}^{\rm (II)} 4 P\bigl(X > X_\star, Y > Y_\star\ \bigm|\ J_{1111}^{\rm (II)}\bigr) - p_{11} + p_{11}^{\rm (I)} + p_{10}^{\rm (I)}\bigr\} \\
    &\quad + p_{11}^{\rm (II)} p_{11}^\star \bigl\{P\bigl(X < X_\star\ |\ J_{1111}^{\rm (II)}\bigr) - P\bigl(X > X_\star\ |\ J_{1111}^{\rm (II)}\bigr)\bigr\} \\
    &\quad + p_{10}^{\rm (II)} p_{11}^\star \bigl\{P\bigl(X < X_\star\ |\ J_{1011}^{\rm (II)}\bigr) - P\bigl(X > X_\star\ |\ J_{1011}^{\rm (II)}\bigr)\bigr\} \\
    &\quad + p_{11} p_{11}^\star \bigl\{ P(Y < Y_\star\ |\ J_{1111}) - P(Y > Y_\star\ |\ J_{1111})\bigr\} \\
    &\quad + p_{01} p_{11}^\star \bigl\{P(Y < Y_\star\ |\ J_{0111}) - P(Y > Y_\star\ |\ J_{1111})\bigr\} \\
    &\quad+ p_{10}^\star \bigl\{p_{11}^{\rm (I)} \bigl[2 \pi_{4, M}^{\rm (I)} - 1\bigr] + p_{11}^{\rm (II)}\bigr\} + p_{11} p_{01}^\star \{2 \pi_{3, M} - 1\} \\
    &\quad + p_{00} p_{11}^\star + p_{11} p_{00}^\star - p_{01} p_{10}^\star - p_{10} p_{01}^\star.
\end{align*}
This expression can be simplified by first computing
\begin{align*}
    p_{11}^\star p_{11}^{\rm (II)} &P(X < X_\star\ |\ J_{1111}^{\rm (II)}) + p_{11}^\star p_{10}^{\rm (II)} P(X < X_\star\ |\ J_{1011}^{\rm (II)}) \\
    &\quad = P(X < X_\star, Y > 0, F(X) > p_2, X_\star > 0, Y_\star > 0) \\
    &\qquad + P(X < X_\star, Y = 0, F(X) > p_2, X_\star > 0, Y_\star > 0) \\
    &\quad = p_{11}^\star\ P(F(X) > p_2)\  P(X < X_\star\ |\ F(X) > p_2, X_\star > 0, Y_\star > 0) \\
    &\quad = p_{11}^\star (1 - p_2)/2.
\end{align*}
By a similar argument, it follows that
\begin{align*}
    p_{11}^\star p_{11}^{\rm (II)} P(X > X_\star\ |\ J_{1111}^{\rm (II)}) + p_{11}^\star p_{10}^{\rm (II)} P(X > X_\star\ |\ J_{1011}^{\rm (II)}) = p_{11}^\star (1 - p_2)/2.
\end{align*}
Observe that the two terms cancel out. \cite{Denuit2017bounds} also showed that $\{Y_\star\ |\ X_\star > 0, Y_\star > 0\}$ can attain any value larger than zero, this specifically follows from $p_{01}^\star = 0$. As a result, we may analogously find that
\begin{align*}
    &p_{11} p_{11}^\star \Bigl\{P(Y < Y_\star\ |\ J_{1111}) - P(Y > Y_\star\ |\ J_{1111})\Bigr\} + p_{01} p_{11}^\star \Bigl\{P(Y < Y_\star\ |\ J_{0111}) - P(Y > Y_\star\ |\ J_{0111})\Bigr\} \\
    &\quad = p_{11}^\star (1 - p_2) \Bigl\{P(Y < Y_\star\ |\ Y > 0, Y_\star > 0, X_\star > 0) - P(Y > Y_\star\ |\ Y > 0, X_\star > 0, Y_\star > 0)\Bigr\} \\
    &\quad = 0.
\end{align*}
Now substituting $p_{11}^\star = 1 - p_2$, $p_{00}^\star = p_1$, $p_{10}^\star = p_2 - p_1$ and $p_{01}^\star = 0$ yields
\begin{align}\label{eq:Q.HM}
    Q(J, M) &= (1 - p_2) \Bigl\{ p_{11}^{\rm (II)} \bigl[4 P\bigl(X > X_\star, Y > Y_\star\ |\ J_{1111}^{\rm (II)}\bigr) - 1\bigr] + p_{10}^{\rm (I)}\Bigr\} \notag\\
    &\quad + (p_2 - p_1) \Bigl\{p_{11}^{\rm (I)} \bigl(2 \pi_{4, M}^{\rm (I)} - 1\bigr) + p_{11}^{\rm (II)}\Bigr\} + p_{00} (1 - p_2) + p_{11} p_1 - p_{01} (p_2 - p_1).
\end{align}

\subsubsection{Proof of Theorem~\ref{thm:conc_lower}}
We now turn to the case where $J_\star$ amounts to the lower Fr\'echet-Hoeffding copula $W$ defined by $W(u, v) = \max\{u + v - 1, 0\}$ for all $u, v \in [0, 1]$. We need to make a case distinction between $p_1 + p_2 \leq 1$ and $p_1 + p_2 > 1$.

First, consider the former. We proceed similarly to the upper Fr\'echet-Hoeffding copula and recall the set of attainable values for the variables distributed as $W$. In particular, we have 
\begin{align*}
    p_1 < F(\{X_\star\ |\ X_\star > 0, Y_\star > 0\}) \leq 1 - p_2, \quad
    F(\{X_\star\ |\ X_\star > 0, Y_\star = 0\}) > 1 - p_2, \\
    p_2 < G(\{Y_\star\ |\ Y_\star > 0, X_\star > 0\}) \leq 1 - p_1, \quad
    G(\{Y_\star\ |\ Y_\star > 0, X_\star = 0\}) > 1 - p_1.
\end{align*}
Therefore, we will distinguish all cases where $X$ and $Y$ satisfy these conditions as well. We now write the terms from the generic expression of concordance in Eq.~\eqref{eq:conc_zid_general} as
\begin{align*}
    p_{11} p_{11}^\star Q(H_1, W)_{11} &= p_{11} p_{11}^\star \Bigl\{4 P(X > X_\star, Y > Y_\star\ |\ J_{1111}) - 1 + \\
    &\qquad\qquad + P(X < X_\star\ |\ J_{1111}) - P(X > X_\star\ |\ J_{1111}) \\
    &\qquad\qquad + P(Y < Y_\star\ |\ J_{1111}) - P(Y > Y_\star\ |\ J_{1111})\Bigr\} \\
    &= p_{11}^\star \Bigl\{4 p_{11}^{\rm (BI')} 4 P\bigl(X > X_\star, Y > Y_\star\ |\ J_{1111}^{\rm (BI')}\bigr) + 4 p_{11}^{\rm (AII')} - p_{11} \\
    &\qquad\qquad + 4 p_{11}^{\rm (AI')} P\bigl(X > X_\star\ |\ J_{1111}^{\rm (AI')}\bigr) + 4 p_{11}^{\rm (BII')} P\bigl(Y > Y_\star\ |\ J_{1111}^{\rm (BII')}\bigr) \\
    &\qquad\qquad + p_{11}^{\rm (I')} \bigl[ P\bigl(X < X_\star\ |\ J_{1111}^{\rm (I')}\bigr) - P\bigl(X > X_\star\ |\ J_{1111}^{\rm (I')}\bigr) \bigr] - p_{11}^{\rm (II')} \\
    &\qquad\qquad + p_{11}^{\rm (B)} \bigl[ P\bigl(Y < Y_\star\ |\ J_{1111}^{\rm (B)}\bigr) - P\bigl(Y > Y_\star\ |\ J_{1111}^{\rm (B)}\bigr) \bigr] - p_{11}^{\rm (A)} \Bigr\},
\end{align*}
and
\begin{align*}
    p_{11} p_{01}^\star \{2 \pi_{3, W} - 1\} &= p_{01}^\star \Bigl\{ p_{11}^{\rm (A)} \bigl[P\bigl(Y > Y_\star\ |\ J_{1101}^{\rm (A)}\bigr) - P\bigl(Y < Y_\star\ |\ J_{1101}^{\rm (A)}\bigr)\bigr] - p_{11}^{\rm (B)} \Bigr\}, \\
    p_{11} p_{10}^\star \{2 \pi_{4, W} - 1\} &= p_{10}^\star \Bigl\{p_{11}^{\rm (II')} \bigl[P\bigl(X > X_\star\ |\ J_{1110}^{\rm (II')}\bigr) - P\bigl(X < X_\star\ |\ J_{1110}^{\rm (II')}\bigr)\bigr] - p_{11}^{\rm (I')} \Bigr\}, \\
    p_{10} p_{11}^\star \{2 \pi_{5, W} - 1\} &= p_{11}^\star \Bigl\{ p_{10}^{\rm (I')} \bigl[ P\bigl(X < X_\star\ |\ J_{1011}^{\rm (I')}\bigr) - P\bigl(X > X_\star\ |\ J_{1011}^{\rm (I')}\bigr)\bigr] - p_{10}^{\rm (II')}\Bigr\}, \\
    p_{01} p_{11}^\star \{2 \pi_{6, W} - 1\} &= p_{11}^\star \Bigl\{p_{01}^{\rm (B)} \bigl[P\bigl(Y < Y_\star\ |\ J_{0111}^{\rm (B)}\bigr) - P\bigl(Y > Y_\star\ |\ J_{0111}^{\rm (B)}\bigr)\bigr] - p_{01}^{\rm (A)}\Bigr\}.
\end{align*}
Here we recognize several probabilities that have the same characteristic in the inequality sense, but are conditioned on the variables being in different regions. By similar arguments as for the upper Fr\'echet-Hoeffding copula, we can simplify these and find that
\begin{align*}
    &p_{11}^\star \Bigl\{p_{11}^{\rm (I')} \bigl[P\bigl(X < X_\star\ |\ J_{1111}^{\rm (I')}\bigr) - P\bigl(X > X_\star\ |\ J_{1111}^{\rm (I')}\bigr)\bigr] \\
    &\qquad + p_{10}^{\rm (I')} \bigl[P\bigl(X < X_\star\ |\ J_{1011}^{\rm (I')}\bigr) - P\bigl(X > X_\star\ |\ J_{1011}^{\rm (I')}\bigr)\bigr] \Bigr\} = 0, \\
    &p_{11}^\star \Bigl\{p_{11}^{\rm (B)} \bigl[ P\bigl(Y < Y_\star\ |\ J_{1111}^{\rm (B)}\bigr) - P\bigl(Y > Y_\star\ |\ J_{1111}^{\rm (B)}\bigr)\bigr] \\
    &\qquad + p_{01}^{\rm (B)} \bigl[ P\bigl(Y < Y_\star\ |\ J_{0111}^{\rm (B)}\bigr) - P\bigl(Y > Y_\star\ |\ J_{0111}^{\rm (B)}\bigr) \bigr] \Bigr\} = 0.
\end{align*}
Now we are in the position to combine all of these expressions. Moreover, using $p_{11}^\star = 1 - p_1 - p_2$, $p_{10}^\star = p_2$, $p_{01}^\star = p_1$, $p_{00}^\star = 0$, $p_{11} = 1 - p_1 - p_2 - p_{00}$, $p_{11}^{\rm (A)} + p_{01}^{\rm (A)} = p_1$ and $p_{11}^{\rm (II')} + p_{10}^{\rm (II')} = p_2$, it follows that
\begin{align*}
    Q(J, W) &= p_{11}^\star \Bigl\{4 p_{11}^{\rm (BI')} P\bigl(X > X_\star, Y > Y_\star\ |\ J_{1111}^{\rm (BI')}\bigr) + 4 p_{11}^{\rm (AII')} - p_{11} - p_{11}^{\rm (II')} - p_{11}^{\rm (A)} - p_{01}^{\rm (A)} - p_{10}^{\rm (II')} \\
    &\qquad\qquad + 4 p_{11}^{\rm (AI')} P\bigl(X > X_\star\ |\ J_{1111}^{\rm (AI')}\bigr) + 4 p_{11}^{\rm (BII')} P\bigl(Y > Y_\star\ |\ J_{1111}^{\rm (BII')}\bigr) \Bigr\} \\
    &\quad + p_{01}^\star \Bigl\{p_{11}^{\rm (A)} \bigl[P\bigl(Y > Y_\star\ |\ J_{1101}^{\rm (A)}\bigr) - P\bigl(Y < Y_\star\ |\ J_{1101}^{\rm (A)}\bigr)\bigr] - p_{11}^{\rm (B)} \Bigr\} \\
    &\quad + p_{10}^\star \Bigl\{ p_{11}^{\rm (II')} \bigl[P\bigl(X > X_\star\ |\ J_{1110}^{\rm (II')}\bigr) - P\bigl(X < X_\star\ |\ J_{1110}^{\rm (II')}\bigr)\bigr] - p_{11}^{\rm (I')} \Bigr\} \\
    &\quad + p_{00} p_{11}^\star + p_{11} p_{00}^\star - p_{01} p_{10}^\star - p_{10} p_{01}^\star \\
    &= (1 - p_1 - p_2) \Bigl\{4 p_{11}^{\rm (BI')} P\bigl(X > X_\star, Y > Y_\star\ |\ J_{1111}^{\rm (BI')}\bigr) + 4 p_{11}^{\rm (AII')} - (1 - p_1 - p_2 + p_{00}) \\
    &\qquad\qquad\qquad\qquad - p_1 - p_2 + 4 p_{11}^{\rm (AI')} P\bigl(X > X_\star\ |\ J_{1111}^{\rm (AI')}\bigr) + 4 p_{11}^{\rm (BII')} P\bigl(Y > Y_\star\ |\ J_{1111}^{\rm (BII')}\bigr) \Bigr\} \\
    &\quad + p_1 \Bigl\{p_{11}^{\rm (A)} \bigl[2 \pi_{3, W}^{\rm (A)} - 1\bigr] - p_{11}^{\rm (B)} \Bigr\} \\
    &\quad + p_2 \Bigl\{ p_{11}^{\rm (II')} \bigl[2 \pi_{4, W}^{\rm (II')} - 1\bigr] - p_{11}^{\rm (I')} \Bigr\} \\
    &\quad + p_{00} (1 - p_1 - p_2) - p_{01} p_2 - p_{10} p_1 \\
    &= (1 - p_1 - p_2) \Bigl\{4 p_{11}^{\rm (BI')} P\bigl(X > X_\star, Y > Y_\star\ |\ J_{1111}^{\rm (BI')}\bigr) - 1 + 4 p_{11}^{\rm (AII')} \\
    &\qquad\qquad\qquad\qquad + 4 p_{11}^{\rm (AI')} P\bigl(X > X_\star\ |\ J_{1111}^{\rm (AI')}\bigr) + 4 p_{11}^{\rm (BII')} P\bigl(Y > Y_\star\ |\ J_{1111}^{\rm (BII')}\bigr) \Bigr\} \\
    &\quad + p_1 \Bigl\{ 2 p_{11}^{\rm (A)} \pi_{3, W}^{\rm (A)} - (1 - p_1) \Bigr\} + p_2 \Bigl\{ 2 p_{11}^{\rm (II')} \pi_{4, W}^{\rm (II')} - (1 - p_2) \Bigr\}.
\end{align*}

Finally, we consider the case where $p_1 + p_2 > 1$. Under these conditions, $p_{11}^\star = 0$, and therefore we do not have to consider more refined regions as before. Instead, simply substituting $p_{10}^\star = 1 - p_1$, $p_{01}^\star = 1 - p_2$ and $p_{00}^\star = p_1 + p_2 - 1$ directly yields
\begin{align*}
    Q(J, W) &= p_{11} (1 - p_2) \{2 \pi_{3, W} - 1\} + p_{11} (1 - p_1) \{2 \pi_{4, W} - 1\} \\
    &\quad + p_{11} (p_1 + p_2 - 1) - p_{01} (1 - p_1) - p_{10} (1 - p_2).
\end{align*}

\subsubsection{Proof of Theorem~\ref{thm:rho}}
Finally, we consider Spearman's rho, i.e., the case where $H_2$
coincides with independence. In particular, we have
\begin{align*}
    Q(J, \Pi) &= P((X_1 - X_2)(Y_1 - Y_3) > 0) - P((X_1 - X_2)(Y_1 - Y_3) < 0)
\end{align*}
where $(X_1, Y_1), (X_2, Y_2)$ and $(X_2, Y_3)$ are independent copies of $(X, Y)$. We again proceed with the general formulation for $Q(J, J_\star)$ as derived in Eq.~\eqref{eq:conc_zid_general}. We now distinguish the cases where $X_3$ and $Y_2$ are zero or strictly positive as well. First, we rewrite each of the terms in Eq.~\eqref{eq:conc_zid_general} as
\begin{align*}
    p_{11} p_{11}^\star Q(J, \Pi)_{11} &= p_{11}^3 Q(J, \Pi)_{11} + p_{11}^2 p_{10} Q(J, \Pi)_{10} + p_{11}^2 p_{01} Q(J, \Pi)_{01} + p_{11} p_{10} p_{01} Q(J, \Pi)_{00}
\end{align*}
and
\begin{align*}
    p_{10} p_{11}^\star P(X_1 < X_2\ |\ J_{1011}) &= p_{10} p_{11} (1 - p_2) P(X_1 < X_2\ |\ J_{1011}, Y_2 > 0) \\
    &\quad + p_{10}^2 (1 - p_2) P(X_1 < X_2\ |\ J_{1011}, Y_2 = 0) \\
    &= p_{10} p_{11} (1 - p_2) \pi_{4, J} + \frac{1}{2} p_{10}^2 (1 - p_2) \\
    p_{01} p_{11}^\star P(Y_1 < Y_3\ |\ J_{0111}) &= p_{01} p_{11} (1 - p_1) P(Y_1 < Y_3\ |\ J_{0111}, X_3 > 0) \\
    &\quad + p_{01}^2 (1 - p_1) P(Y_1 < Y_3\ |\ J_{0111}, X_3 = 0) \\
    &= p_{01} p_{11} (1 - p_1) \pi_{3, J} + \frac{1}{2} p_{01}^2 (1 - p_1) \\
    p_{11} p_{10}^\star P(X_1 > X_2\ |\ J_{1110}) &= p_{11}^2 p_2 P(X_1 > X_2\ |\ J_{1110}, Y_2 > 0) \\
    &\quad + p_{11} p_{10} p_2 P(X_1 > X_2\ |\ J_{1110}, Y_2 = 0) \\
    &= \frac{1}{2} p_{11}^2 p_2 + p_{11} p_{10} p_2 \pi_{4, J} \\
    p_{11} p_{01}^\star P(Y_1 > Y_2\ |\ J_{1101}) &= p_{11}^2 p_1 P(Y_1 > Y_3\ |\ J_{1101}, X_3 > 0) \\
    &\quad + p_{11} p_{01} p_1 P(Y_1 > Y_3\ |\ J_{1101}, X_3 > 0) \\
    &= \frac{1}{2} p_{11}^2 p_1 + p_{11} p_{01} p_1 \pi_{3, J},
\end{align*}
with similar derivations for the other terms. We now obtain
\begin{align*}
    Q(J, \Pi) &= p_{11}^3 Q(J, \Pi)_{11} + p_{11}^2 p_{10} Q(J, \Pi)_{10} \\
    &\quad + p_{11}^2 p_{01} Q(J, \Pi)_{01} + p_{11} p_{10} p_{01} Q(J, \Pi)_{00} \\
    &\quad +  \left(p_{11} p_{10} (1 - p_2) \pi_{4, J} + \frac{1}{2} p_{10}^2 (1 - p_2)\right) - \left(p_{11} p_{10} (1 - p_2) \overline \pi_{4, J} + \frac{1}{2} p_{10}^2 (1 - p_2)\right) \\
    &\quad + \left(p_{01} p_{11} (1 - p_1) \pi_{3, J} + \frac{1}{2} p_{01}^2 (1 - p_1)\right) - \left(p_{01} p_{11} (1 - p_1) \overline \pi_{3, J} + \frac{1}{2} p_{01}^2 (1 - p_1)\right) \\
    &\quad + \left(\frac{1}{2} p_{11}^2 p_2 + p_{11} p_{10} p_2 \pi_{4, J}\right) - \left(\frac{1}{2} p_{11}^2 p_2 + p_{11} p_{10} p_2 \overline \pi_{4, J}\right) \\
    &\quad + \left(\frac{1}{2} p_{11}^2 p_1 + p_{11} p_{01} p_1 \pi_{3, J}\right) - \left(\frac{1}{2} p_{11}^2 p_1 + p_{11} p_{01} p_1 \overline \pi_{3, J}\right) \\
    &\quad + p_{00} p_{11}^\star + p_{11} p_{00}^\star - p_{01} p_{10}^\star - p_{10} p_{01}^\star.
\end{align*}
Here we define $\overline \pi_{3, J} = 1 - \pi_{3, J}$ and $\overline \pi_{4, J} = 1 - \pi_{4, J}$. Besides, \cite{Pimentel2009kendall} noted that $p_{00} p_{11}^\star + p_{11} p_{00}^\star - p_{01} p_{10}^\star - p_{10} p_{01}^\star = p_{11} p_{00} - p_{10} p_{01}$. Therefore, the expression can be simplified and Spearman's rho can be given by
\begin{align*}
    \rho &= 3 p_{11}^3 Q(J, \Pi)_{11} + 3 p_{11}^2 p_{10} Q(J, \Pi)_{10} + 3 p_{11}^2 p_{01} Q(J, \Pi)_{01} + 3 p_{11} p_{10} p_{01} Q(J, \Pi)_{00} \\
    &\quad + 3 p_{11} \Bigl\{p_{10} (2 \pi_{4, J} - 1) + p_{01} (2 \pi_{3, J} - 1)\Bigr\} + 3 (p_{11} p_{00} - p_{10} p_{01}).
\end{align*}

\subsection{Proof of~\hyperref[thm:boundsgamma]{Theorem~\ref{thm:boundsgamma}}}
In this section, we derive the attainable bounds on Gini's gamma and Spearman's footrule using the expressions that we derived in the previous section. Recall that the bounds are attained under the upper and lower Fr\'echet-Hoeffding copulas. \cite{Denuit2017bounds} derived sharp upper and lower bounds for Kendall's tau by computing $Q(M, M)$ and $Q(W, W)$ explicitly for zero-inflated data. Specifically, they derived that
\begin{align*}
    Q(M, M) &= 1 - \max\{p_1, p_2\}^2, & Q(W, W) &= \max\{0, 1 - p_1 - p_2\}^2 - 2 (1 - p_1)(1 - p_2).
\end{align*}
On the other hand, Theorem~\ref{thm:rho} provides us with expressions for $Q(\Pi, M)$ and $Q(\Pi, W)$. Therefore, the only term that remains to be computed is $Q(W, M)$. We make use of the expression for $Q(W, M)$ as given in Eq.~\eqref{eq:Q.HM}. We consider three case distinctions under the general assumption that $p_1 \leq p_2$.

\begin{itemize}
    \item We first consider the case where $p_1 + p_2 > 1$, then $p_{11} = 0$ and therefore $p_{11}^{\rm (I)} = p_{11}^{\rm (II)} = 0$ as well. Besides, $p_{00} = p_1 + p_2 - 1$ and $p_{01} = 1 - p_2$. Finally, we compute
    \begin{align*}
        p_{10}^{\rm (I)} &= P(X > 0, Y = 0, F(X) \leq p_2) \\
        &= P(F(X) \leq p_2, Y = 0) - P(X = 0, Y = 0) \\
        &= \max\{p_2 + p_2 - 1, 0\} - \max\{p_1 + p_2 - 1, 0\} \\
        &= p_2 - p_1.
    \end{align*}
    Substitution of these probabilities directly into Eq.~\eqref{eq:Q.HM} yields
    \begin{align*}
        Q(W, M) &= (1 - p_2) (p_2 - p_1) + (p_1 + p_2 - 1) (1 - p_2) - (1 - p_2) (p_2 - p_1) \\
        &= (p_2 + p_1 - 1) (1 - \max\{p_1, p_2\}).
    \end{align*}

    \item Secondly, assume that $p_1 + p_2 \leq 1$ and $p_2 \leq \frac{1}{2}$. We again start by computing the individual probability masses. Recall that from \citep{Denuit2017bounds}, we already have $p_{11} = 1 - p_1 - p_2$ and $p_{01} = p_1$. The remaining probabilities can be computed by
    \begin{align*}
        p_{11}^{\rm (I)} &= P(X > 0, Y > 0, F(X) \leq p_2) \\
        &= P(0 < X, F(X) \leq p_2) - P(F(X) \leq p_2, Y = 0) + P(X = 0,  Y = 0) \\
        &= (p_2 - p_1) - \max\{p_2 + p_2 - 1, 0\} + \max\{p_1 + p_2 - 1, 0\} \\
        &= p_2 - p_1, \\
        p_{11}^{\rm (II)} &= p_{11} - p_{11}^{\rm (I)} = 1 - 2 p_2, \\
        p_{10}^{\rm (I)} &= P(X > 0, Y = 0, F(X) \leq p_2) \\
        &= P(F(X) \leq p_2, Y = 0) - P(X = 0, Y = 0) \\
        &= \max\{p_2 + p_2 -1 , 0\} - \max\{p_1 + p_2 - 1, 0\} \\
        &= 0.
    \end{align*}
    We now turn to computing the remaining probabilities that involve inequalities between the pairs $(X, Y)$ and $(X_\star, Y_\star)$. In particular, we will first find an expression for $P \bigl(X > X_\star, Y > Y_\star\ |\ J_{1111}^{\rm (II)}\bigr)$. Given $J_{1111}^{\rm (II)}$, the values of $X$ and $Y$ are restricted such that $0 < X \leq F^{-1}(1 - p_2)$ and $0 < Y \leq G^{-1}(1 - p_2)$. We can derive an expression for the probability in question by conditioning on whether these conditions are also satisfied for $X_2$ and $Y_2$. As a result, we obtain
    \begin{align*}
        &4 p_{11}^\star p_{11}^{\rm (II)} P\bigl(X > X_\star, Y > Y_\star\ \bigm |\ J_{1111}^{\rm (II)}\bigr) \\
        &\quad= 4 (1 - 2 p_2) p_{11}^{\rm (II)} P\bigl( X > X_\star, Y > Y_\star\ \bigm |\ J_{1111}^{\rm (II)}, F(X_\star) \leq 1 - p_2\bigr) \\
        &\quad = (1 - 2 p_2)^2.
    \end{align*}
    Finally, observe that given $J_{1110}^{\rm (I)}$, the random variables $X$ and $X_\star$ have the same continuous distribution, from which it follows that
    \begin{equation*}
        \pi_{4, M}^{\rm (I)} = P\bigl(X > X_\star\ \bigm|\ J_{1110}^{\rm (I)}\bigr) = \frac{1}{2}.
    \end{equation*}
    Substituting these results into Eq.~\eqref{eq:Q.HM} yields
    \begin{align*}
        Q(W, M) &= (1 - 2 p_2)^2 - (1 - p_2) (1 - 2 p_2) \\
        &\quad + (p_2 - p_1) (1 - 2 p_2) + (1 - p_1 - p_2) p_1 - p_1 (p_2 - p_1) \\
        &= 0.
    \end{align*}

    \item Now consider the last case, where $p_1 + p_2 \leq 1$ and $p_2 > \frac{1}{2}$. Under these conditions, we still have $p_{11} = 1 - p_1 - p_2$ and $p_{01}  = p_1$, however
    \begin{align*}
        p_{11}^{\rm (I)} &= P(X > 0, Y > 0, F(X) \leq p_2) \\
        &= P(0 < X, F(X) \leq p_2) - P(F(X) \leq p_2, Y = 0) + P(X = 0, Y = 0) \\
        &= (p_2 - p_1) - \max\{p_2 + p_2 - 1, 0\} + \max\{p_1 + p_2 - 1, 0\} \\
        &= 1 - p_1 - p_2, \\
        p_{11}^{\rm (II)} &= p_{11} - p_{11}^{\rm (I)} = 0, \\
        p_{10}^{\rm (I)} &= P(X > 0, Y = 0, F(X) \leq p_2) \\
        &= P(F(X) \leq p_2, Y = 0) - P(X = 0, Y = 0) \\
        &= \max\{p_2 + p_2 - 1, 0\} - \max\{p_1 + p_2 - 1, 0\} \\
        &= 2 p_2 - 1.
    \end{align*}
    Finally, we compute $\pi_{4, M}^{\rm (I)}$. Given $J_{1110}^{\rm (I)}$, the variable $X_1$ satisfies $0 < F(X) \leq 1 - p_2$, therefore we condition on $X_\star$ satisfying this condition as well. Doing so, we obtain
    \begin{align*}
        p_{11}^\star p_{11}^{\rm (I)} \pi_{4, M}^{\rm (I)} &= (1 - p_1 - p_2) p_{11}^{\rm (I)} P(X > X_\star\ |\ J_{1111}^{\rm (I)}, F(X_\star) \leq 1 - p_2) \\
        &\quad + (2 p_2 - 1) p_{11}^{\rm (I)} P(X > X_\star\ |\ J_{1111}^{\rm (I)}, F(X_\star) > 1 - p_2) \\
        &= \frac{1}{2} (1 - p_1 - p_2)^2.
    \end{align*}
    It now follows that
    \begin{align*}
        Q(W, M) &= (1 - p_2) (2 p_2 - 1) \\
        &\quad + (1 - p_1 - p_2)^2 - (p_2 - p_1) (1 - p_1 - p_2) \\
        &\quad + (1 - p_1 - p_2) p_1 - p_1 (p_2 - p_1) \\
        &= 0.
    \end{align*}
\end{itemize}
Now that we have expressions available for $Q(M, M)$, $Q(W, M)$, $Q(W, W)$, $Q(M, \Pi)$ and $Q(W, \Pi)$, we can formulate the bounds on Gini's gamma and Spearman's footrule. Suppose that $p_1 + p_2 \leq 1$, then
\begin{align*}
    \gamma_{\max} &= Q(M, M) + Q(W, M) - Q(\Pi, M) - Q(\Pi, W) \\
    &= \bigl(1 - \max\{p_1, p_2\}^2\bigr) - \frac{1}{3} \bigl(1 - \max\{p_1, p_2\}^3\bigr) - \frac{1}{3} \bigl(p_1^3 + p_2^3 - 1\bigr) \\
    &= 1 - \max\{p_1, p_2\}^2 - \frac{1}{3} \min\{p_1, p_2\}^3,
\end{align*}
and
\begin{align*}
    \gamma_{\min} &= Q(W, M) + Q(W, W) - Q(\Pi, M) - Q(\Pi, W) \\
    &= \bigl((1 - p_1 - p_2)^2 - 2 (1 - p_1)(1 - p_2)\bigr) - \frac{1}{3} \bigl(1 - \max\{p_1, p_2\}^3\bigr) - \frac{1}{3} \bigl(p_1^3 + p_2^3 - 1\bigr) \\
    &= (1 - p_1 - p_2)^2 - 2 (1 - p_1)(1 - p_2) - \frac{1}{3} \min\{p_1, p_2\}^3
\end{align*}
On the other hand, when $p_1 + p_2 > 1$, we conclude that
\begin{align*}
    \gamma_{\max} &= Q(M, M) + Q(W, M) - Q(\Pi, M) - Q(\Pi, W) \\
    &= \bigl(1 - \max\{p_1, p_2\}^2\bigr) + (p_1 + p_2 - 1) (1 - \max\{p_1, p_2\}) \\
    &\qquad - \frac{1}{3} \bigl(1 - \max\{p_1, p_2\}^3\bigr) + (1 - p_1)(1 - p_2) \\
    &= \frac{2}{3} + \max\{p_1, p_2\} \bigl(1 - 2 \max\{p_1, p_2\} + \frac{1}{3} \max\{p_1, p_2\}^2\bigr),
\end{align*}
and
\begin{align*}
    \gamma_{\min} &= Q(W, M) + Q(W, W) - Q(\Pi, M) - Q(\Pi, W) \\
    &= (p_1 + p_2 - 1)(1 - \max\{p_1, p_2\}) - 2 (1 - p_1)(1 - p_2) \\
    &\qquad - \frac{1}{3} \bigl(1 - \max\{p_1, p_2\}^3\bigr) + (1 - p_1)(1 - p_2) \\
    &= (p_1 + p_2 + \min\{p_1, p_2\} - 2) (1 - \max\{p_1, p_2\}) - \frac{1}{3} \bigl(1 - \max\{p_1, p_2\}^3\bigr).
\end{align*}
Finally, we can compute the bounds on Spearman's footrule directly by noting that
\begin{align*}
    \phi_{\max} &= \frac{3}{2} Q(M, M) - \frac{3}{2} Q(M, \Pi) \\
    &= \frac{3}{2} \bigl(1 - \max\{p_1, p_2\}^2\bigr) - \frac{1}{2} \bigl(1 - \max\{p_1, p_2\}^3\bigr) \\
    &= 1 - \frac{3}{2} \max\{p_1, p_2\}^2 + \frac{1}{2} \max\{p_1, p_2\}^3,
\end{align*}
and
\begin{align*}
    \phi_{\min} &= \frac{3}{2} Q(W, M) - \frac{3}{2} Q(M, \Pi) \\
    &= \frac{3}{2} \max\{0, p_1 + p_2 - 1\} (1 - \max\{p_1, p_2\}) - \frac{1}{2} \bigl(1 - \max\{p_1, p_2\}^3\bigr).
\end{align*}

\subsection{Proof of~\hyperref[thm:sper-rho-bounds]{Theorem~\ref{thm:sper-rho-bounds}}}
In this section, we proof Theorem~\ref{thm:sper-rho-bounds} and compute the sharp upper and lower bounds on Spearman's rho.

We start with the upper bound and recall that it is attained $(X, Y)$ are joined through the upper Fr\'echet-Hoeffding copula. Without loss of generality, assume that $p_1 \leq p_2$, then \cite{Denuit2017bounds} noted that $p_{11} = 1 - p_2$, $p_{10} = p_2 - p_1$, $p_{01} = 0$ and $p_{00} = p_1$. As discussed before, they also derived that $F(\{X\,|\, X > 0, Y = 0\}) \leq p_2$ and $F(\{X\,|\, X > 0, Y > 0\}) > p_2$. In particular, this suggests that $\{X\,|\, X > 0, Y = 0\} < \{X\,|\, X > 0, Y > 0\}$ with probability one, and hence $\pi_{5, J} = 1$. Moreover, this allows us to compute $Q(J, \Pi)_{10}$ as follows. First, note that given $S_{10}$, the random variables $X_1$ and $X_2$ are distributed as $(X\, |\, X > 0, Y > 0)$ and $(X\, |\, X > 0, Y = 0)$ respectively, and therefore $X_2 < X_1$ with probability one. Moreover, $Y_1$ and $Y_3$ have the same continuous distribution, and hence
\begin{align*}
    Q(J, \Pi)_{10} &= P\bigl((X_1 - X_2)(Y_1 - Y_3) > 0\ |\ S_{10}\bigr) - P\bigl((X_1 - X_2)(Y_1 - Y_3) < 0\ |\ S_{10}\bigr) \\
    &= P(Y_1 > Y_3\ |\ S_{10}) - P(Y_1 < Y_3\ |\ S_{10}) \\
    &= 0.
\end{align*}
Now using $3 Q(J, \Pi)_{11} = 1$ since this amounts to Spearman's rho for continuous random variables under the upper Fr\'echet-Hoeffding copula, we conclude that
\begin{align*}
    \rho_{\max} &= (1 - p_2)^3 + 3 (1 - p_2) p_1 + 3 (1 - p_2) (p_2 - p_1) \\
    &= \bigl(1 - 3 p_2 + 3 p_2^2 - p_2^3\bigr) + 3 (1 - p_2) p_2 \\
    &= 1 - p_2^3.
\end{align*}
We will now derive the lower bound in a similar manner. First, consider the case where $p_1 + p_2 \leq 1$, then \cite{Denuit2017bounds} noted that $p_{11} = 1 - p_1 - p_2$, $p_{10} = p_2$, $p_{01} = p_1$ and $p_{00} = 0$. Moreover, given $S_{10}$, the random variables $X_1$ and $X_2$ are distributed as $(X\, |\, X > 0, Y > 0)$ and $(X\, |\, X > 0, Y = 0)$ respectively, and satisfy $X_1 < X_2$ with probability one. It now follows that $\pi_{4, J} = 0$ and
\begin{align*}
    Q(J, \Pi)_{10} &= P\bigl((X_1 - X_2)(Y_1 - Y_3) > 0\ |\ S_{10}\bigr) - P\bigl((X_1 - X_2)(Y_1 - Y_3) < 0\ |\ S_{10}\bigr) \\
    &= P(Y_1 < Y_3\ |\ S_{10}) - P(Y_1 > Y_3\ |\ S_{10}) \\
    &= 0.
\end{align*}
By a similar reasoning, it follows that $Q(J, \Pi)_{01} = 0$. Finally, given $S_{00}$, the random variables $X_1$ and $X_2$ are distributed as given $S_{10}$, however now $Y_3$ has the same distribution as $(Y\,|\, X = 0, Y > 0)$. \cite{Denuit2017bounds} argued that $G(\{Y\ |\ X > 0, Y > 0\}) \geq 1 - p_1$ and $G\{(Y\ |\ X = 0, Y > 0\}) < 1 - p_1$ and so $\pi_{3, J} = 0$. It follows that
\begin{align*}
    Q(J, \Pi)_{00} &= P\bigl((X_1 - X_2)(Y_1 - Y_3) > 0\ |\ S_{00}\bigr) - P\bigl((X_1 - X_2)(Y_1 - Y_3) < 0\ |\ S_{00}\bigr) \\
    &= P(X_1 > X_2, Y_1 > Y_3\ |\ S_{00}) + P(X_1 < X_2, Y_1 < Y_3\ |\ S_{00}) \\
    &\quad - P(X_1 > X_2, Y_1 < Y_3\ |\ S_{00}) - P(X_1 < X_2, Y_1 > Y_3\ |\ S_{00}) \\
    &= 1.
\end{align*}
Now substituting $3 Q(J, \Pi)_{11} = -1$ and the rest yields
\begin{align*}
    \rho_{\min} &= - (1 - p_1 - p_2)^3 + 3 (1 - p_1 - p_2) p_1 p_2 - 3 p_1 p_2 - 3 (1 - p_1 - p_2) (p_1 + p_2) \\
    &= \bigl(p_1^3 + p_2^3 - 3 p_1^2 + 3 p_1^2 p_2 + 3 p_1 - 6 p_1 p_2 - 3 p_2^2 + 3 p_2 + 3 p_1 p_2^2 - 1\bigr) \\
    &\qquad + \bigl(3 p_1 p_2 -3 p_1^2 p_2 - 3 p_1 p_2^2\bigr) - 3 p_1 p_2 - 3 p_1 - 3 p_2 + 3 p_1^2 + 3 p_2^2 + 6 p_1 p_2 \\
    &= p_1^3 + p_2^3 - 1.
\end{align*}
Now consider the case where $p_1 + p_2 > 1$. Then, we have $p_{11} = 0$, $p_{10} = 1 - p_1$, $p_{01} = 1 - p_2$ and $p_{00} = p_1 + p_2 - 1$, and so
\begin{align*}
    \rho_{\min} &= - 3 (1 - p_1)(1 - p_2).
\end{align*}

\section*{Supplementary material: Alternative proofs of the results} \label{sup}
\renewcommand{\thesubsection}{S.\arabic{subsection}} 

\section*{Summary}
In the following, we present alternative proofs of Theorems 3.4 and Theorem 3.5. Specifically, in Section~\ref{sup:mixed}, we present another method of computing $Q(M, W)$, which leads to the main result of Theorem 3.4. Sections~\ref{sup:spm_rho1} and~\ref{sup:spm_rho2} are devoted to alternative proofs of the upper and lower bounds on Spearman's rho. These are based on the general expression of the concordance function from \cite{Mesfioui2005properties} and the general formulation of the bounds by \cite{Mesfioui2022rho} respectively.

\subsection{Alternative proof of Theorem 3.4 }\label{sup:mixed}
This section explores an alternative method to compute $Q(M, W)$, the remainder of the proof of Theorem 3.1 is the same as in the Appendix of the main manuscript. Following the results from ~\cite{Mesfioui2005properties}, we get
		\begin{align}\label{E1}
			Q(M, W)&= \EE\left[\min\{F(X_\star),G(Y_\star)\}\right] +\EE\left[\min\{F(X^-_\star),G(Y^-_\star)\}\right]  \nonumber \\
			&\quad + \EE\left[\min\{F(X^-_\star),G(Y_\star)\}\right] +\EE\left[\min\{F(X_\star),G(Y^-_\star)\}\right] -1,
		\end{align}
		where $X_\star$ and $Y_\star$ have joint distribution function $W$ and are distributed as $X$ and $Y$, respectively. We will compute each of these terms separately in several case distinctions involving $p_1$ and $p_2$. 
		\begin{itemize}
			\item In a first time, let us consider the cases where $p_1+p_2\ge 1$. We have
			\begin{equation*}
				p_{00}=p_1+p_2-1,\quad p_{01}=1-p_2, \quad p_{10}=1-p_1,\quad p_{11}=0.
			\end{equation*}
			Starting with the easy computing one,
			\begin{align}\label{E2}
				\EE\left[\min\{F(X^-_\star),G(Y^-_\star)\}\right]&=\EE\left[\min\{F(X^-_\star),G(Y^-_\star)\}| X_\star>0,Y_\star>0\right]p_{11} \nonumber  \\
				&=0. 
			\end{align}
			Similarly,  $\EE\left[\min\{F(X^-_\star),G(Y_\star)\}\right]$ is computed by conditioning upon $[X_\star>0,Y_\star=0]$ or $[X_\star>0,Y_\star>0]$. It follows that
			\begin{align}\label{E3}
				\EE\left[\min\{F(X^-_\star),G(Y_\star)\}\right]&=\EE\left[\min\{U,p_2\}|U>p_1\right](1-p_1)+0. 
			\end{align}
			Likewise, $\EE\left[\min\{F(X_\star),G(Y^-_\star)\}\right]$ is computed by conditioning upon $[X_\star=0,Y_\star>0]$ or $[X_\star>0,Y_\star>0]$, so that
			\begin{align} \label{E4}
				\EE\left[\min\{F(X_\star),G(Y^-_\star)\}\right]&=\EE\left[\min\{p_1, U\}|U>p_2\right](1-p_2)+0.
			\end{align}
			Finally, the last term $\EE\left[\min\{F(X_\star),G(Y_\star)\}\right]$ is calculated under the cases $[X_\star=0,Y_\star=0]$ or $[X_\star=0,Y_\star>0]$ or $[X_\star>0,Y_\star=0]$ or $[X_\star>0,Y_\star>0]$, which leads to 
			\begin{align} \label{E5}
				\EE\left[\min\{F(X_\star),G(Y_\star)\}\right]&= \min\{p_1,p_2\}(p_1+p_2-1)  \nonumber \\
				&\quad+\EE\left[\min\{p_1,U\}|U>p_2\right](1-p_2)  \nonumber \\
				&\quad+ \EE\left[\min\{U,p_2\}|U>p_1\right](1-p_1)+0.
			\end{align}
Further, one has
			\begin{align*}
				\EE\left[\min\{p_1,U\}|U>p_2\right](1-p_2)=\begin{cases}
					p_1(1-p_2),\quad & p_1\le p_2,\\
					\frac{2p_1-p_1^2-p_2^2}{2},\quad &p_1 > p_2,
				\end{cases}
			\end{align*}
			and
			\begin{align*}
				\EE\left[\min\{U,p_2\}|U>p_1\right](1-p_1)=\begin{cases}
					\frac{2p_2-p_2^2-p_1^2}{2} ,\quad & p_1\le p_2,\\
					p_2(1-p_1),\quad &p_1 > p_2.
				\end{cases}
			\end{align*}
			
			Consequently, 
			\begin{equation*}
				Q(M, W)=
				-p_1p_2+\min\{p_1, p_2\}-\max\{p_1^2, p_2^2\}+2\max\{p_1, p_2\}-1
			\end{equation*}
			and
			\item  In a second time, let us consider the cases where $p_1+p_2 < 1$. Hence, 
			\begin{equation*}
				p_{00}=0,\quad p_{01}=p_1, \quad p_{10}=p_2,\quad p_{11}=1-p_1-p_2.
			\end{equation*}
			Starting with the easy computing one,
			\begin{align*}
				\EE\left[\min\{F(X^-_\star),G(Y^-_\star)\}\right]&=\EE\left[\min\{F(X^-_\star),G(Y^-_\star)\}| X_\star>0,Y_\star>0\right]p_{11} \\
				&=\EE\left[\min\{U,1-U\}|p_1 < U < 1-p_2\right](1-p_1-p_2).
			\end{align*}
			Similarly,  $\EE\left[\min\{F(X^-_\star),G(Y_\star)\}\right]$ is computed by conditioning upon $[X_\star>0,Y_\star=0]$ or $[X_\star>0,Y_\star>0]$. It follows that
			\begin{align}\label{E3}
				\EE\left[\min\{F(X^-_\star),G(Y_\star)\}\right]&=\EE\left[\min\{U,p_2\}|U>1-p_2\right]p_2\nonumber\\
				&\quad +\EE\left[\min\{U,1-U\}|p_1 < U < 1-p_2\right](1-p_1-p_2). 
			\end{align}
			Likewise, $\EE\left[\min\{F(X_\star),G(Y^-_\star)\}\right]$ is computed by conditioning upon $[X_\star=0,Y_\star>0]$ or $[X_\star>0,Y_\star>0]$, so that
			\begin{align} \label{E4}
				\EE\left[\min\{F(X_\star),G(Y^-_\star)\}\right]&=\EE\left[\min\{p_1,1-U\}|U\le p_1\right]p_1\nonumber \\
				&\quad + \EE\left[\min\{U,1-U\}|p_1 < U < 1-p_2\right](1-p_1-p_2). 
			\end{align}
			Finally, the last term $\EE\left[\min\{F(X_\star),G(Y_\star)\}\right]$ is calculated under the cases $[X_\star=0,Y_\star=0]$ or $[X_\star=0,Y_\star>0]$ or $[X_\star>0,Y_\star=0]$ or $[X_\star>0,Y_\star>0]$, which leads to 
			\begin{align} \label{E5}
				\EE\left[\min\{F(X_\star),G(Y_\star)\}\right]&=\EE\left[\min\{p_1,1-U\}|U\le p_1\right]p_1  \nonumber \\
				&\quad+ \EE\left[\min\{U,p_2\}|U>1-p_2\right]p_2\nonumber \\
				&\quad +\EE\left[\min\{U,1-U\}|p_1 < U < 1-p_2\right](1-p_1-p_2). 
			\end{align}
	Further, one has
			\begin{align*}
				\EE\left[\min\{p_1,1-U\}|U\le p_1\right]p_1=\begin{cases}
					p_1(1-p_1)+(2p_1-1)-\frac{p_1^2-(1-p_1)^2}{2},\quad & p_1 \ge 1/2,\\
					p_1^2, \quad & p_1< 1/2,
				\end{cases}
			\end{align*}
			and
			\begin{align*}
				\EE\left[\min\{U,p_2\}|U>1-p_2\right]p_2=\begin{cases}
					\frac{p_2^2-(1-p_2)^2}{2}+p_2(1-p_2) ,\quad & p_2\ge 1/2,\\
					p_2^2,\quad & p_2 < 1/2.
				\end{cases}
			\end{align*}
			Moreover,
			\begin{align*}
				\EE\left[\min\{U,1-U\}|p_1 < U < 1-p_2\right](1-p_1-p_2)= \begin{cases}
					\frac{(1-p_2)^2-p_1^2}{2}, \quad & p_2\ge 1/2,\\
					\frac{3-2(p_1^2+1+p_2^2)}{4},\quad & p_2 < 1/2, p_1 < 1/2,\\
					\frac{(1-p_1)^2-p_2^2}{2}, \quad & p_2 < 1/2, p_1 \ge 1/2.
				\end{cases} 
			\end{align*}
			Now since $p_1 + p_2 \leq 1$, the combination $p_1, p_2 \geq 1/2$ is not possible, and consequently, $Q(M, W) = 0$. \qed
		\end{itemize}

\subsection{Alternative proof of Theorem 3.5}\label{sup:spm_rho1}
Here, we present another proof of Theorem 3.5 revolving around the general expression of the concordance function as noted by \cite{Mesfioui2005properties}. Specifically, we get
\begin{align*}
			Q(M, \Pi)&= \EE\left[\min\{F(X_\star),G(Y_\star)\}\right] +\EE\left[\min\{F(X^-_\star),G(Y^-_\star)\}\right]  \nonumber \\
			&\quad + \EE\left[\min\{F(X^-_\star),G(Y_\star)\}\right] +\EE\left[\min\{F(X_\star),G(Y^-_\star)\}\right] -1,
\end{align*}
where $X_\star$ and $Y_\star$ are independent and distributed as $X$ and $Y$, respectively.  Starting with the easy computing one:
\begin{eqnarray}\label{eq:2025-06-27, 10:38AM}
\EE\left[\min\{F(X^-_\star),G(Y^-_\star)\}\right]&=&\EE\left[\min\{F(X^-_\star),G(Y^-_\star)\}| X_\star>0,Y_\star>0\right](1-p_1)(1-p_2) \nonumber  \\
&=&\EE\left[\min\{p_1+(1-p_1)U,p_2+(1-p_2)V\}\right](1-p_1)(1-p_2), 
\end{eqnarray}
where the random $U$ and $V$ are independent and follow uniform distribution on [0, 1]. Similarly,  $\EE\left[\min\{F(X^-_\star),G(Y^-_\star)\}\right]$ is computed by conditioning upon $[X_\star>0,Y_\star=0]$ or $[X_\star>0,Y_\star>0]$. It follows that:
\begin{eqnarray*}
\EE\left[\min\{F(X^-_\star),G(Y_\star)\}\right]&=&\EE\left[\min\{p_1+(1-p_1)U,p_2\}\right](1-p_1)p_2 \nonumber \\
&&+\EE\left[\min\{p_1+(1-p_1)U,p_2+(1-p_2)V\}\right](1-p_1)(1-p_2). 
\end{eqnarray*}
Likewise, $\EE\left[\min\{F(X^-_\star),G(Y^-_\star)\}\right]$ is computing by conditioning upon $[X_\star=0,Y_\star>0]$ or $[X_\star>0,Y_\star>0]$. So that,
\begin{eqnarray} \label{eq:2025-06-27, 10:41AM}
\EE\left[\min\{F(X_\star),G(Y^-_\star)\}\right]&=&\EE\left[\min\{p_1,p_2+(1-p_2)V\}\right]p_1(1-p_2)  \nonumber \\
&&+\EE\left[\min\{p_1+(1-p_1)U,p_2+(1-p_2)V\}\right](1-p_1)(1-p_2).  
\end{eqnarray}
Finally, the last term $\EE\left[\min\{F(X_\star),G(Y_\star)\}\right]$ is computing under the fourth cases $[X_\star=0,Y_\star=0]$ or $[X_\star=0,Y_\star>0]$ or $[X_\star>0,Y_\star=0]$ or $[X_\star>0,Y_\star>0]$, which leads to 
\begin{eqnarray} \label{eq:2025-06-27, 10:40AM}
\EE\left[\min\{F(X_\star),G(Y_\star)\}\right]&=& \min\{p_1,p_2\}p_1p_2  \nonumber \\
&&+\EE\left[\min\{p_1,p_2+(1-p_2)V\}\right]p_1(1-p_2)  \nonumber \\
&&+ \EE\left[\min\{p_1+(1-p_1)U,p_2\}\right](1-p_1)p_2  \nonumber \\
&&+\EE\left[\min\{p_1+(1-p_1)U,p_2+(1-p_2)V\}\right](1-p_1)(1-p_2). 
\end{eqnarray}
Futrher, one has 
\begin{eqnarray*}
\EE\left[\min\{p_1,p_2+(1-p_2)V\}\right]p_1(1-p_2)=\begin{cases}
	p_1^2(1-p_2),\quad & p_1\le p_2,\\
\frac{2p_1^2-p_1^3-p_1p_2^2}{2},\quad &p_1 > p_2.
\end{cases}
\end{eqnarray*}
\begin{eqnarray*}
\EE\left[\min\{p_1+(1-p_1)U,p_2\}\right](1-p_1)p_2=\begin{cases}
		\frac{2p_2^2-p_2^3-p_1^2p_2}{2} ,\quad & p_1\le p_2,\\
p_2^2(1-p_1),\quad &p_1 > p_2.
\end{cases}
\end{eqnarray*}
\begin{align*}
&\EE\left[\min\{p_1+(1-p_1)U,p_2+(1-p_2)V\}\right](1-p_1)(1-p_2)\\
&=\begin{cases}
p_1(1-p_1)(1-p_2) + (1-p_2)\left[(p_2-p_1)-\frac{p_2^2-p_1^2}{2}+\frac{(1-p_2)^2}{3}\right], \quad  &p_1\le p_2,\\[10pt] 
p_2(1-p_1)(1-p_2) + (1-p_1)\left[(p_1-p_2)-\frac{p_1^2-p_2^2}{2}+\frac{(1-p_1)^2}{3}\right], \quad  &p_1> p_2,
\end{cases}
\end{align*}
Inserting the above results into equations \eqref{eq:2025-06-27, 10:38AM}-\eqref{eq:2025-06-27, 10:40AM}. We obtain that
\begin{equation*}
	\rho_{\max}= 3Q(M, \Pi)= 1- \max\{p_1^3, p_2^3\}.
\end{equation*}
Similarly, for the lower bound, one has 
\begin{align}\label{eq:2025-01-13, 6:41PM}
	Q(W, \Pi)&=\EE\left[\max\{F(X_\star)+G(Y_\star)-1, 0\}\right] +\EE\left[\max\{F(X^-_\star)+G(Y^-_\star)-1, 0\}\right] \nonumber \\
	&\quad + \EE\left[\max\{F(X^-_\star)+G(Y_\star)-1, 0\}\right] +\EE\left[\max\{F(X_\star)+G(Y^-_\star)-1, 0\}\right]-1.
\end{align}
We obtain that, for the four terms,
\begin{align*}
&\EE\left[\max\{F(X_\star)+G(Y_\star)-1, 0\}\right]\\
 &\quad  =  \max\{p_1+p_2-1, 0\}p_1p_2\\
&\qquad +\EE\left[ \max\{p_1+p_2+(1-p_2)V-1, 0\}\right]p_1(1-p_2)\\
&\qquad +\EE\left[\max\{p_1+(1-p_1)U+p_2-1, 0\}\right](1-p_1)p_2\\
&\qquad +\EE\left[\max\{p_1+(1-p_1)U+p_2(1-p_2)V-1, 0\}\right](1-p_1)(1-p_2).
\end{align*}
And
\begin{align*}
& \EE\left[\max\{F(X^-_\star)+G(Y^-_\star)-1, 0\}\right]\\
 &\quad =\EE\left[\max\{p_1+(1-p_1)U+p_2+(1-p_2)V-1, 0\}\right](1-p_1)(1-p_2).
\end{align*}
And
\begin{align*}
\EE\left[\max\{F(X^-_\star)+G(Y_\star)-1, 0\}\right]&  = \EE\left[\max\{p_1+(1-p_1)U+p_2-1, 0\}\right] (1-p_1)p_2\\
	& \quad + \EE\left[\max\{p_1+(1-p_1)U+p_2+(1-p_2)V-1, 0\}\right](1-p_1)(1-p_2).
\end{align*}
And
\begin{equation*}\begin{split}
	\EE\left[\max\{F(X_\star)+G(Y^-_\star)-1, 0\}\right]  &= \EE\left[\max\{p_1+p_2+(1-p_2)V-1, 0\}\right] p_1(1-p_2)\\
	& \quad + \EE\left[\max\{p_1+(1-p_1)U+p_2+(1-p_2)V-1, 0\}\right](1-p_1)(1-p_2).
\end{split}\end{equation*}
Further, one has 
\begin{align*}
&\EE\left[ \max\{p_1+p_2+(1-p_2)V-1, 0\}\right]p_1(1-p_2)\\
 &\quad = \begin{cases}
 	\frac{p_1^3}{2}, \quad &p_1+p_2\le 1,\\
 	\frac{p_1(1-p_2)(2p_1+p_2-1)}{2},\quad &p_1+p_2> 1
 \end{cases}
\end{align*}

\begin{align*}
&\EE\left[\max\{p_1+(1-p_1)U+p_2-1, 0\}\right] (1-p_1)p_2 \\
 &\quad = \begin{cases}
 	\frac{p_2^3}{2}, \quad &p_1+p_2\le 1,\\
 	\frac{(1-p_1)p_2(p_1+2p_2-1)}{2},\quad &p_1+p_2> 1
 \end{cases}
\end{align*}

\begin{align*}
& \EE\left[\max\{p_1+(1-p_1)U+p_2+(1-p_2)V-1, 0\}\right](1-p_1)(1-p_2) \\
 &\quad =  \begin{cases}
 	\frac{(p_1+p_2)(1-p_1)(1-p_2)}{2}+\frac{(1-p_1-p_2)^3}{6}, \quad &p_1+p_2\le 1,\\
 	\frac{(1-p_1)(1-p_2)(p_1+p_2)}{2},\quad &p_1+p_2> 1
 \end{cases}
\end{align*}
Finally, we have
\begin{equation}\label{eq:2025-01-14, 9:40AM}
	\rho_{\min}= 3Q(W, \Pi)=\begin{cases}
		p_1^3+p_2^3-1,\quad &p_1+p_2 \le 1,\\
		-3(1-p_1)(1-p_2), \quad &p_1+p_2> 1.
	\end{cases}
\end{equation}
This completes the proof.\qed

\subsection{Another alternative proof of Theorem 3.5}\label{sup:spm_rho2}
In this section, we present another proof of Theorem 3.5. According to \cite{Mesfioui2022rhodisc}, for two discrete random variables $X$ and $Y$ with distribution functions $F$ and $G$ respectively, the upper bound on Spearman's rho can be given by
\begin{align}\label{eq:spm_disc_max}
    \rho_{\max} = 9 &- 3 \bigl( \mathbb E\bigl[G(Y) \left(F(\phi(Y)) + F(\phi(Y)-)\right)\bigr] + \mathbb E\bigl[G(Y-) \left(F(\phi(Y-)) + F(\phi(Y-)-)\right)\bigr]\bigr) \notag\\
    & - 3 \bigl( \mathbb E\bigl[F(X) \left(G(\psi(X)) + G(\psi(X)-)\right)\bigr] + \mathbb E \bigl[F(X-) \left(G(\psi(X-)) + F(\psi(X-)-)\right)\bigr]\bigr).
\end{align}
Here, we define the functions
\begin{align*}
    \phi(y) = \min\{x \geq 0\,:\, F(x) \geq G(y)\}, \quad y \geq 0, \\
    \psi(x) = \min\{y \geq 0\,:\, G(y) > F(x)\}, \quad x \geq 0,
\end{align*}
where $\psi(x) = +\infty$ if $F(x) = 1$. Without loss of generality, assume that $p_1 \leq p_2$. We first compute $\phi(y)$ for any $y \geq 0$. Substituting the distribution functions of $X$ and $Y$ for $y \geq 0$, this becomes
\begin{align*}
    \phi(y) &= \min\{x \geq 0\,:\, p_1 + (1 - p_1) \tilde F(x) \geq p_2 + (1 - p_2) \tilde G(y)\} \\
    &= \min\{x \geq 0\,:\, \tilde F(x) \geq \frac{1}{1 - p_1} \left(p_2 + (1 - p_2) \tilde G(y)\right)\} \\
    &= \tilde F^{-1}\left(\frac{1}{1 - p_1} \left(p_2 - p_1 + (1 - p_2) \tilde G(y)\right)\right).
\end{align*}
The latter follows from $p_2 \geq p_1$ and $\tilde F$ being continuous. We can now evaluate $F(\phi(y))$ as
\begin{align*}
    F(\phi(y)) &= p_1 + (1 - p_1) \tilde F\left( \tilde F^{-1} \left(\frac{1}{1 - p_1} \left(p_2 - p_1 + (1 - p_2) \tilde G(y)\right)\right)\right) \\
    &= p_1 + p_2 - p_1 + (1 - p_2) \tilde G(y) \\
    &= G(y).
\end{align*}
Therefore, it now follows that
\begin{align*}
    &\mathbb E[G(Y) F(\phi(Y))] = \mathbb E[G^2(Y)] \\
    &\qquad= p_2 G^2(0) + (1 - p_2) \int_0^\infty G^2(y) \diff \tilde G(y) \\
    &\qquad= p_2^3 + (1 - p_2) \int_0^\infty (p_2 + (1 - p_2) \tilde G(y))^2 \diff \tilde G(y) \\
    &\qquad= p_2^3 + (1 - p_2) p_2^2 \int_0^\infty 1 \diff \tilde G(y) + 2 (1 - p_2) p_2 \int_0^\infty \tilde G(y) \diff \tilde G(y) + (1 - p_2)^3 \int_0^\infty \tilde G^2(y) \diff \tilde G(y) \\
    &\qquad= p_2^3 + (1 - p_2) p_2^2 + (1 - p_2) p_2 + \frac{1}{3} (1 - p_2)^3 \\
    &\qquad= \frac{1}{3} + \frac{2}{3} p_2^3 
\end{align*}
Following the same derivations, we find that
\begin{align*}
    &\mathbb E[G(Y-) F(\phi(Y-))] = (1 - p_2) \int_0^\infty G(y-) F(\phi(y-)) \diff \tilde G(y) \\
    &\qquad= (1 - p_2) \int_0^\infty G^2(y) \diff \tilde G(y) \\
    &\qquad= (1 - p_2) p_2 + \frac{1}{3} (1 - p_2)^3 \\
    &\qquad= \frac{1}{3} - \frac{1}{3} p_2^3.
\end{align*}
Since $\phi(y-) = \phi(y)$ and $F(\phi(y)-) = F(\phi(y))$ for all $y \geq 0$, we now find that $\mathbb E[G(Y) F(\phi(Y-))] = \mathbb E[G(Y) F(\phi(Y))]$ and $\mathbb E[G(Y-) F(\phi(Y-)-)] = \mathbb E[G(Y-) F(\phi(Y-))]$.

We proceed similarly to compute $\mathbb E[F(X) G(\psi(X))]$ and the remaining terms. First, we derive the values of $\psi(x)$ by
\begin{align*}
    \psi(x) &= \min\{y \geq 0\,:\, G(y) \geq F(x)\} \\
    &= \min\{y \geq 0\,:\, p_2 + (1 - p_2) \tilde G(y) > p_1 + (1 - p_1) \tilde F(x)\} \\
    &= \min\bigl\{y \geq 0\,:\, \tilde G(y) > \frac{1}{1 - p_2} \left(p_1 - p_2 + (1 - p_1) \tilde F(x)\right)\bigr\}.
\end{align*}
Now let $\tilde x' \geq 0$ such that $p_1 - p_2 + (1 - p_1) \tilde F(x') = 0$, then
\begin{align*}
    \psi(x) &= \begin{cases}
        0 & \text{if } x \leq \tilde x', \\
        \tilde G^{-1}\left(\frac{1}{1 - p_2} \left(p_1 - p_2 + (1 - p_1) \tilde F(x)\right)\right) & \text{if } x > \tilde x.
    \end{cases}
\end{align*}
Based on this expression, we first consider the simplification
\begin{align*}
    p_2 + (1 - p_2) \tilde G\left(\frac{1}{1 - p-2} \left(p_1 - p_2 + (1 - p_1) \tilde F(x)\right)\right) = p_1 + (1 - p_1) \tilde F(x) = F(x)
\end{align*}
such that
\begin{align*}
    G(\psi(x)) = \begin{cases}
        p_2 & \text{if } x \leq \tilde x', \\
        F(x) & \text{if } x > \tilde x'
    \end{cases} \quad \text{and} \quad G(\psi(x)-) = \begin{cases}
        0 & \text{if } x \leq \tilde x', \\
        F(x) & \text{if } x > \tilde x'.
    \end{cases}
\end{align*}
Let us first compute the following integrals. Using $(1 - p_1)\tilde F(x') = p_2 - p_1$, one may find that
\begin{align*}
    &(1 - p_1) \int_0^{x'} F(x) \diff \tilde F(x) \\
    &\qquad= (1 - p_1) p_1 \int_0^{x'} 1 \diff \tilde F(x) + (1 - p_1)^2 \int_0^{x'} \tilde F(x) \diff \tilde F(x) \\
    &\qquad= (1 - p_1) p_1 \tilde F(x') + \frac{1}{2} (1 - p_1)^2 \tilde F^2(x') \\
    &\qquad= p_1 (p_2 - p_1) + \frac{1}{2} (p_2 - p_1)^2, \\
    &(1 - p_1) \int_{x'}^\infty F^2(x) \diff \tilde F(x) \\
    &\qquad= (1 - p_1) p_1^2 \int_{x'}^\infty 1 \diff \tilde F(x) + 2 (1 - p_1)^2 p_1 \int_{x'}^\infty \tilde F(x) \diff \tilde F(x) + (1 - p_1)^3 \int_{x'}^\infty \tilde F^2(x) \diff \tilde F(x) \\
    &\qquad= (1 - p_1) p_1^2 (1 - \tilde F(x')) + (1 - p_1)^2 p_1 (1 - \tilde F^2(x')) + \frac{1}{3} (1 - p_1)^3 (1 - \tilde F^3(x')) \\
    &\qquad= p_1^2 \bigl((1 - p_1) - (p_2 - p_1)\bigr) + p_1 \bigl((1 - p_1)^2 - (p_2 - p_1)^2\bigr) + \frac{1}{3} \bigl((1 - p_1)^3 - (p_2 - p_1)^3\bigr) \\
    &\qquad= \frac{1}{3} - \frac{1}{3} p_2^3.
\end{align*}
It now follows that
\begin{align*}
    &\mathbb E[F(X) G(\psi(X))] \\
    &\qquad= p_1^2 p_2 + (1 - p_1) \int_0^{x'} F(x) p_2 \diff \tilde F(x) + (1 - p_1) \int_{x'}^\infty F^2(x) \diff \tilde F(x) \\
    &\qquad= p_1^2 p_2 + p_1 p_2 (p_2 - p_1) + \frac{1}{2} p_2 (p_2 - p_1)^2 + \frac{1}{3} - \frac{1}{3} p_2^3 \\
    &\qquad= \frac{1}{3} + \frac{1}{2} p_1^2 p_2 + \frac{1}{6} p_2^3, \\
    &\mathbb E[F(X) G(\psi(X)-)] = (1 - p_1) \int_{x'}^\infty F^2(x) \diff \tilde F(x) \\
    &\qquad= (1 - p_1) \int_{x'}^\infty \left(p_1 + (1 - p_1) \tilde F(x)\right)^2 \diff \tilde F(x) \\
    &\qquad= \frac{1}{3} - \frac{1}{3} p_2^3.
\end{align*}
Since $\phi(x-) = \phi(x)$ and $F(0-) = 0$, it follows that
\begin{align*}
    &\mathbb E[F(X-) G(\psi(X-))] \\
    &\qquad= (1 - p_1) \int_0^{x'} F(x) p_2 \diff \tilde F(x) + (1 - p_1) \int_{x'}^\infty F^2(x) \diff \tilde F(x) \\
    &\qquad= p_1 p_2 (p_2 - p_1) + \frac{1}{2} p_2 (p_2 - p_1)^2 + \frac{1}{3} - \frac{1}{3} p_2^3.
\end{align*}
Now since $\psi(x-) = \psi(x) = 0$ and so $G(\psi(x-)-) = 0$ for $x \leq x'$, we find that
\begin{align*}
    \mathbb E[F(X-) G(\psi(X-)-)] = (1 - p_1) \int_{x'}^\infty F^2(x) \diff \tilde F(x) = \frac{1}{3} - \frac{1}{3} p_2^3.
\end{align*}
We combine all of these results and insert them into Eq.~\eqref{eq:spm_disc_max} to conclude
\begin{align*}
    \rho_{\max} &= 1 - p_2^3.
\end{align*}

We now turn to the lower bound on $\rho$. According to \cite{Mesfioui2022rhodisc}, this can be given by
\begin{align}\label{eq:spm_disc_min}
    \rho_{\min} &= 3 \bigl( \mathbb E\bigl[G(Y) \left(\bar F(\tilde \phi(Y)) + \bar F(\tilde \phi(Y)-)\right)\bigr] + \mathbb E\bigl[G(Y-) \left(\bar F(\tilde \phi(Y-)) + \bar F(\tilde \phi(Y-)-)\right)\bigr] \bigr) \notag\\
    &\qquad - 3 \bigl( \mathbb E\bigl[\bar F(X) \left(\bar G(\tilde \psi(X)) + \bar G(\tilde \psi(X)-)\right)\bigr] + \mathbb E\bigl[\bar F(X-) \left(\bar G(\tilde \psi(X-)) + \bar G(\tilde \psi(X-)-)\right)\bigr] \bigr) \\
    &\qquad + 3 \bigl( \mathbb E\bigl[\bar G(\tilde \psi(X-)-)] - \mathbb E[\bar F(\tilde \psi(Y))] - 1 \bigr), \notag
\end{align}
where we define the functions
\begin{align*}
    \tilde \phi(y) = \min\{x \geq 0\,:\, F(x) \geq \bar G(y)\}, \quad y \geq 0, \\
    \tilde \psi(x) = \min\{y \geq 0\,:\, G(y) \geq \bar F(x)\}, \quad x \geq 0,
\end{align*}
and $\bar F(x) = 1 - F(x)$ and $\bar G(y) = 1 - G(y)$ for all $x, y \geq 0$. We make a case distinction, and first consider the case where $p_1 + p_2 > 1$. Then, for any $x, y \geq 0$, we have $(1 - p_1) \tilde F(x) \geq 1 - p_1 - p_2 - (1 - p_2) \tilde G(y)$ and so
\begin{align*}
    \tilde \phi(y) &= \min\{x \geq 0\,:\, F(x) \geq 1 - G(y)\} \\
    &= \min\{x \geq 0\,:\, (1 - p_1) \tilde F(x) \geq 1 - p_1 - p_2 - (1 - p_2) \tilde G(y)\} \\
    &= 0,
\end{align*}
for any $y \geq 0$, and therefore $\tilde \phi(y-) = \tilde \phi(y)$ for $y > 0$ as well. Besides, $G(0-) = 0$ and hence $\tilde \phi(0-) = \infty$. Now, recalling that $\bar F(x) = 1 - F(x)$, we have $\bar F(\tilde \phi(y)) = (1 - p_1) \mathbbm 1(y \geq 0)$ and $\bar F(\tilde \phi(y)-) = \mathbbm 1(y \geq 0)$, where $\mathbbm 1$ is the indicator function. Before computing the terms in Eq.~\eqref{eq:spm_disc_min} involving $\tilde \phi$, we first evaluate the integral
\begin{align*}
    \int_0^\infty G(y) \diff \tilde G(y) &= \int_0^\infty p_2 + (1 - p_2) \tilde G(y) \diff \tilde G(y) \\
    &= p_2 + (1 - p_2) \int_0^\infty \tilde G(y) \diff \tilde G(y) \\
    &= p_2 + \frac{1}{2} (1 - p_2) \\
    &= \frac{1}{2} (1 + p_2).
\end{align*}
This can be used to derive expressions for the terms involving $\tilde \phi$ as
\begin{align*}
    &\mathbb E[G(Y) (\bar F(\tilde \phi(Y)) + \bar F(\tilde \phi(Y)-))] \\
    &\qquad= p_2 G(0) \bigl(\bar F(\tilde \phi(0)) + \bar F(\tilde \phi(0)-)\bigr) + (1 - p_2) \int_0^\infty G(y) \bigl(\bar F(\tilde \phi(y)) + \bar F(\tilde \phi(y)-)\bigr) \diff \tilde G(y) \\
    &\qquad= p_2^2 (1 - p_1 + 1) + (1 - p_2) \int_0^\infty G(y) (1 - p_1 + 1) \diff \tilde G(y) \\
    &\qquad= p_2^2 (2 - p_1) + \frac{1}{2} (1 - p_2) (2 - p_1) (1 + p_2), \\
    &\mathbb E[G(Y-) (\bar F(\tilde \phi(Y-)) + \bar F(\tilde \phi(Y-)-))] \\
    &\qquad= p_2 G(0-) \bigl(\bar F(\tilde \phi(0-)) + \bar F(\tilde \phi(0-)-)\bigr) + (1 - p_2) \int_0^\infty G(y-) \bigl(\bar F(\tilde \phi(y-)) + \bar F(\tilde \phi(y-)-)\bigr) \diff \tilde G(y) \\
    &\qquad= (1 - p_2) \int_0^\infty G(y) (1 - p_1 + 1) \diff \tilde G(y) \\
    &\qquad= \frac{1}{2} (1 - p_2) (2 - p_1) (1 + p_2), \\
    &\mathbb E[\bar F(\tilde \phi(Y))] \\
    &\qquad= p_2 \bar F(\tilde \phi(0)) + (1 - p_2) \int_0^\infty \bar F(\tilde \phi(y)) \diff \tilde G(y) \\
    &\qquad= p_2 (1 - p_1) + (1 - p_2) (1 - p_1) \\
    &\qquad= 1 - p_1.
\end{align*}
By the same argument, it follows that $\tilde \psi(x) = 0$ for $x \geq 0$ and $\tilde \psi(x) = \infty$ for $x < 0$. Therefore, $\bar G(\tilde \psi(x)) = (1 - p_2) \mathbbm 1(x \geq 0)$ and $\bar G(\tilde \psi(x)-) = \mathbbm 1(x \geq 0)$. Following the same computations as before, we first find that
\begin{align*}
    \int_0^\infty \bar F(x) \diff \tilde F(x) &= \int_0^\infty (1 - F(x)) \diff \tilde F(x) \\
    &= 1 - \frac{1}{2} (1 + p_1) \\
    &= \frac{1}{2} (1 - p_1).
\end{align*}
Now we can evaluate the terms involving $\tilde \psi$ as
\begin{align*}
    &\mathbb E[\bar F(X) (\bar G(\tilde \psi(X)) + \bar G(\tilde \psi(X)-))] \\
    &\qquad= p_1 F(0) (\bar G(\tilde \psi(0)) + \bar G(\tilde \psi(0)-)) + (1 - p_1) \int_0^\infty \bar F(x) (\bar G(\tilde \psi(x)) + \bar G(\tilde \psi(x)-))\diff \tilde F(x) \\
    &\qquad= (1 - p_1) p_1 (2 - p_2) + \frac{1}{2} (1 - p_1)^2 (2 - p_2) \\
    &\qquad= \frac{1}{2} (1 - p_1)(2 - p_2)(1 + p_1), \\
    &\mathbb E[\bar F(X-) (\bar G(\tilde \psi(X-)) + \bar G(\tilde \psi(X-)-))] \\
    &\qquad= p_1 F(0-) (\bar G(\tilde \psi(0-)) + \bar G(\tilde \psi(0-)-)) + (1 - p_1) \int_0^\infty \bar F(x-) (\bar G(\tilde \psi(x-)) + \bar G(\tilde \psi(x-)-)) \diff \tilde F(x) \\
    &\qquad= (1 - p_1) \int_0^\infty \bar F(x) (2 - p_2) \diff \tilde F(x) \\
    &\qquad= \frac{1}{2} (1 - p_1)^2 (2 - p_2), \\
    &\mathbb E[\bar G(\tilde \psi(X-)-)] \\
    &\qquad= p_1 \bar G(\tilde \psi(0-)-) + (1 - p_1) \int_0^\infty \bar G(\tilde \psi(X-)-) \diff \tilde F(x) \\
    &\qquad= (1 - p_1) \int_0^\infty 1 \diff \tilde F(x) \\
    &\qquad= (1 - p_1).
\end{align*}
We insert all of these results into Eq.~\eqref{eq:spm_disc_min} to find that
\begin{align*}
    \rho_{\min} &= 3 \left[\left(p_2^2 (2 - p_1) + \frac{1}{2} (1 - p_2)(2 - p_1)(1 + p_2)\right) + \left(\frac{1}{2} (1 - p_2)(2 - p_1) (1 + p_2)\right)\right] \\
    &\qquad - 3 \left[\left(\frac{1}{2} (1 - p_1)(2 - p_2)(1 + p_1)\right) + \left(\frac{1}{2} (1 - p_1)^2(2 - p_2)\right)\right] \\
    &\qquad + 3 [(1 - p_1) - (1 - p_1) - 1\bigr] \\
    &= -3 (1 - p_1)(1 - p_2).
\end{align*}
Now we consider the case where $p_1 + p_2 \leq 1$. As before, we will first compute $\tilde \phi(y)$ for any $y \geq 0$. Specifically, using the definition of zero-inflated distributed random variables, we have
\begin{align*}
    \tilde \phi(y) &= \min\{x \geq 0\,:\, F(x) \geq 1 - G(y)\} \\
    &= \min\{x \geq 0\,:\, (1 - p_1) \tilde F(x) \geq 1 - p_1 - p_2 - (1 - p_2) \tilde G(y)\}
\end{align*}
for $y \geq 0$. Now let $\tilde y \geq 0$ be the point such that $(1 - p_2) \tilde G(y) = 1 - p_1 - p_2$, then
\begin{align*}
    \tilde \phi(y) = \begin{cases}
        \infty & \text{if } y < 0, \\
        \tilde F^{-1}\bigl(\frac{1}{1 - p_1} (1 - p_1 - p_2 - (1 - p_2) \tilde G(y)) \bigr) & \text{if } 0 \leq y \leq \tilde y, \\
        0 & \text{if } y > \tilde y.
    \end{cases}
\end{align*}
We will simplify this expression and evaluate $\bar F$ at $\tilde \phi(y)$ by first rewriting
\begin{align*}
    F(\tilde \phi(y)) &= p_1 + (1 - p_1) \tilde F\bigl( \tilde F^{-1}\bigl(\tfrac{1}{1 - p_1} (1 - p_1 - p_2 - (1 - p_2) \tilde G(y)\bigr)\bigr) \\
    &= 1 - p_2 - (1 - p_2) \tilde G(y) \\
    &= \bar G(y),
\end{align*}
for $0 \leq y \leq \tilde y$, and hence
\begin{align*}
    \bar F(\tilde \phi(y)) &= \begin{cases}
        0 & \text{if } y < 0, \\
        G(y) & \text{if } 0 \leq y < \tilde y, \\
        1 - p_1 & \text{if } \tilde y \leq y,
    \end{cases}, \quad \text{and} \quad \bar F(\tilde \phi(y)-) = \begin{cases}
        0 & \text{if } y < 0, \\
        G(y) & \text{if } 0 \leq y < \tilde y, \\
        1 & \text{if } \tilde y \leq y.
    \end{cases}
\end{align*}
Moreover, by continuity of $\tilde F$ and $\tilde G$, it follows that
\begin{align*}
    \bar F(\tilde \phi(y-)) &= \begin{cases}
        0 & \text{if } y \leq 0, \\
        G(y) & \text{if } 0 < y < \tilde y, \\
        1 - p_1 & \text{if } \tilde y \leq y,
    \end{cases}, \quad \text{and} \quad \bar F(\tilde \phi(y-)-) = \begin{cases}
        0 & \text{if } y \leq 0, \\
        G(y) & \text{if } 0 < y < \tilde y, \\
        1 & \text{if } \tilde y \leq y.
    \end{cases}
\end{align*}
Before we compute the terms involving $\tilde \phi$, we first consider several integrals and expectations that we will frequently employ later. First note that
\begin{align*}
    \int_a^b \tilde G^d(y) \diff \tilde G(y) = \int_{\tilde G(a)}^{\tilde G(b)} u^d du = \frac{\tilde G^{d + 1}(b) - \tilde G^{d + 1}(a)}{d + 1},
\end{align*}
for all $a < b$ and $d \geq 1$. Now, as a consequence, we have
\begin{align*}
    &(1 - p_2) \int_0^{\tilde y} G^2(y) \diff \tilde G(y) \\
    &\qquad= (1 - p_2) \int_0^{\tilde y} (p_2 + (1 - p_2) \tilde G(y))^2 \diff \tilde G(y) \\
    &\qquad= (1 - p_2) p_2^2 \int_0^{\tilde y} 1 \diff \tilde G(y) + 2 (1 - p_2)^2 p_2 \int_0^{\tilde y} \tilde G(y) \diff \tilde G(y) + (1 - p_2)^3 \int_0^{\tilde y} \tilde G^2(y) \diff \tilde G(y) \\
    &\qquad= (1 - p_2) p_2^2 \tilde G(\tilde y) + (1 - p_2)^2 p_2 \tilde G^2(\tilde y) + \frac{1}{3} (1 - p_2)^3 \tilde G^2(\tilde y) \\
    &\qquad= (1 - p_1 - p_2) p_2^2 + (1 - p_1 - p_2)^2 p_2 + \frac{1}{3} (1 - p_1 - p_2)^3 \\
    &\qquad = \frac{1}{3} (1 - p_1^3 - p_2^3) + p_1^2 - p_1, \\
    &(1 - p_2) \int_{\tilde y}^\infty G(y) \diff \tilde G(y) \\
    &\qquad= (1 - p_2) \int_{\tilde y}^\infty p_2 + (1 - p_2) \tilde G(y) \diff \tilde G(y) \\
    &\qquad= (1 - p_2) p_2 \int_{\tilde y}^\infty 1 \diff \tilde G(y) + (1 - p_2)^2 \int_{\tilde y}^\infty \tilde G(y) \diff \tilde G(y) \\
    &\qquad= (1 - p_2) p_2 (1 - \tilde G(\tilde y)) + \frac{1}{2} (1 - p_2)^2 (1 - \tilde G^2(\tilde y)) \\
    &\qquad= (1 - p_2) p_2 - p_2 (1 - p_1 - p_2) + \frac{1}{2} (1 - p_2)^2 - \frac{1}{2} (1 - p_1 - p_2)^2 \\
    &\qquad= p_1 p_2 + \frac{1}{2} (1 - 2 p_2 + p_2^2) - \frac{1}{2} (1 - 2 p_1 - 2 p_2 + 2 p_1 p_2 + p_1^2 + p_2^2) \\
    &\qquad = p_1 - \frac{1}{2} p_1^2.
\end{align*}
Finally, combining these results yields
\begin{align*}
    &\mathbb E[G(Y) (\bar F(\tilde \phi(Y)) + \bar F(\tilde \phi(Y)-))] \\
    &\qquad= p_2 G(0) (\bar F(\tilde \phi(0)) + \bar F(\tilde \phi(0)-)) + (1 - p_2) \int_0^{\tilde y} G(y) (\bar F(\tilde \phi(y)) + \bar F(\tilde \phi(y)-)) \diff \tilde G(y) \\
    &\qquad\qquad + (1 - p_2) \int_{\tilde y}^\infty G(y) (\bar F(\tilde \phi(y)) + \bar F(\tilde \phi(y)-)) \diff \tilde G(y) \\
    &\qquad= 2 p_2^3 + 2 (1 - p_2) \int_0^{\tilde y} G^2(y) \diff \tilde G(y) + (1 - p_2) (2 - p_1) \int_{\tilde y}^\infty G(y) \diff \tilde G(y) \\
    &\qquad= 2 p_2^3 + \frac{2}{3} (1 - p_1^3 - p_2^3) + 2 p_1^2 - 2 p_1 + (2 - p_1) \bigl(p_1 - \frac{1}{2} p_1^2\bigr) \\
    &\qquad= \frac{2}{3} - \frac{1}{6} p_1^3 + \frac{4}{3} p_2^3, \\
    &\mathbb E[G(Y-) (\bar F(\tilde \phi(Y-)) + \bar F(\tilde \phi(Y-)-))] \\
    &\qquad= p_2 G(0-) (\bar F(\tilde \phi(0-)) + \bar F(\tilde \phi(0-)-)) + (1 - p_2) \int_0^{\tilde y} G(y-) (\bar F(\tilde \phi(y-)) + \bar F(\tilde \phi(y-)-)) \diff \tilde G(y) \\
    &\qquad\qquad + (1 - p_2) \int_{\tilde y}^\infty G(y-) (\bar F(\tilde \phi(y-)) + \bar F(\tilde \phi(y-)-)) \diff \tilde G(y) \\
    &\qquad= 2 (1 - p_2) \int_0^{\tilde y} G^2(y) \diff \tilde G(y) + (1 - p_2) (2 - p_1) \int_{\tilde y}^\infty G(y) \diff \tilde G(y) \\
    &\qquad= \frac{2}{3} (1 - p_1^3 - p_2^3) + 2 p_1^2 - 2 p_1 + (2 - p_1) \bigl(p_1 - \frac{1}{2} p_1^2\bigr), \\
    &\mathbb E[\bar F(\tilde \phi(y))] \\
    &\qquad= p_2 \bar F(\tilde \phi(0)) + (1 - p_2) \int_0^{\tilde y} \bar F(\tilde \phi(y)) \diff \tilde G(y) + (1 - p_2) \int_{\tilde y}^\infty \bar F(\tilde \phi(y)) \diff \tilde G(y) \\
    &\qquad= p_2^2 + (1 - p_2) \int_0^{\tilde y} G(y) \diff \tilde G(y) + (1 - p_2) \int_{\tilde y}^\infty (1 - p_1) \diff \tilde G(y) \\
    &\qquad= p_2^2 + (1 - p_2) p_2 \tilde G(\tilde y) + \frac{1}{2} (1 - p_2)^2 \tilde G^2(\tilde y) + (1 - p_1) (1 - p_2) (1 - \tilde G(\tilde y)) \\
    &\qquad= p_2^2 + p_2 (1 - p_1 - p_2) + \frac{1}{2} (1 - p_1 - p_2)^2 + (1 - p_1)(1 - p_2 - (1 - p_1 - p_2)) \\
    &\qquad= (p_1 + p_2) (1 - p_1) + \frac{1}{2} (1 - p_1 - p_2)^2.
\end{align*}
We now turn to the computation of the terms involving $\tilde \psi$. We can follow the same argument as for $\tilde \phi$ to find that
\begin{align*}
    \bar G(\tilde \psi(x)) &= \begin{cases}
        0 & \text{if } x < 0, \\
        F(x) & \text{if } 0 \leq x \leq \tilde x, \\
        1 - p_2 & \text{if } \tilde x < x
    \end{cases}, \quad \text{and} \quad \bar G(\tilde \psi(x)-) = \begin{cases}
        0 & \text{if } x < 0, \\
        F(x) & \text{if } 0 \leq x \leq \tilde x, \\
        1 & \text{if } \tilde x < x,
    \end{cases} \\
    \bar G(\tilde \psi(x-)) &= \begin{cases}
        0 & \text{if } x \leq 0, \\
        F(x) & \text{if } 0 \leq x \leq \tilde x, \\
        1 - p_2 & \text{if } \tilde x < x,
    \end{cases} \quad \text{and} \quad \bar G(\tilde \psi(x-)-) = \begin{cases}
        0 & \text{if } x \leq 0, \\
        F(x) & \text{if } 0 < x \leq \tilde x, \\
        1 & \text{if } \tilde x < x,
    \end{cases}
\end{align*}
where $\tilde x \geq 0$ is the point such that $(1 - p_1) \tilde F(\tilde x) = 1 - p_1 - p_2$. Again, we first consider several general integrals that are relevant for the computation of the terms involving $\tilde \psi$. We find that
\begin{align*}
    &(1 - p_1) \int_0^{\tilde x} \bar F(x) F(x) \diff \tilde F(x)  \\
    &\qquad= (1 - p_1) \int_0^{\tilde x} (1 - p_1 - (1 - p_1) \tilde F(x)) (p_1 + (1 - p_1) \tilde F(x)) \diff \tilde F(x) \\
    &\qquad= (1 - p_1)^2 p_1 \int_0^{\tilde x} 1 \diff \tilde F(x) + (1 - p_1)^2 (1 - 2 p_1) \int_0^{\tilde x} \tilde F(x) \diff \tilde F(x) - (1 - p_1)^3 \int_0^{\tilde x} \tilde F^2(x) \diff \tilde F(x) \\
    &\qquad= (1 - p_1)^2 p_1 \tilde F(\tilde x) + \frac{1}{2} (1 - p_1)^2 (1 - 2p_1) \tilde F^2(\tilde x) - \frac{1}{3} (1 - p_1)^3 \tilde F^3(\tilde x) \\
    &\qquad= (1 - p_1) p_1 (1 - p_1 - p_2) + \frac{1}{2} (1 - 2 p_1) (1 - p_1 - p_2)^2 - \frac{1}{3} (1 - p_1 - p_2)^3, \\
    &(1 - p_1) \int_{\tilde x}^\infty \bar F(x) \diff \tilde F(x) \\
    &\qquad= (1 - p_1) \int_{\tilde x}^\infty (1 - p_1 - (1 - p_1) \tilde F(x)) \diff \tilde F(x) \\
    &\qquad= (1 - p_1)^2 \int_{\tilde x}^\infty 1 \diff \tilde F(x) - (1 - p_1)^2 \int_{\tilde x}^\infty \tilde F(x) \diff \tilde F(x) \\
    &\qquad= (1 - p_1)^2 (1 - \tilde F(\tilde x)) - \frac{1}{2} (1 - p_1)^2 \bigl(1 - \tilde F^2(\tilde x)\bigr) \\
    &\qquad= (1 - p_1) p_2 + \frac{1}{2} (1 - 2 p_1 - 2 p_2 + 2 p_1 p_2 + p_1^2 + p_2^2 - 1 + 2 p_1 - p_1^2) \\
    &\qquad= \frac{1}{2} p_2^2
\end{align*}
Now we can compute the relevant terms as
\begin{align*}
    &\mathbb E[\bar F(X) (\bar G(\tilde \psi(X)) + \bar G(\tilde \psi(X)-))] \\
    &\qquad= p_1 \bar F(0) (\bar G(\tilde \psi(0)) + \bar G(\tilde \psi(0)-)) + (1 - p_1) \int_0^{\tilde x} \bar F(x) (\bar G(\tilde \psi(x)) + \bar G(\tilde \psi(x)-)) \diff \tilde F(x) \\
    &\qquad\qquad + (1 - p_1) \int_{\tilde x}^\infty \bar F(x) (\bar G(\tilde \psi(x)) + \bar G(\tilde \psi(x)-)) \diff \tilde F(x) \\
    &\qquad= 2 p_1^2 (1 - p_1) + 2 (1 - p_1) \int_0^{\tilde x} \bar F(x) F(x) \diff \tilde F(x) + (1 - p_1) (2 - p_2) \int_{\tilde x}^\infty \bar F(x) \diff \tilde F(x) \\
    &\qquad= 2 p_1^2 (1 - p_1) + 2 (1 - p_1) p_1 (1 - p_1 - p_2) \\
    &\qquad\qquad + (1 - 2 p_1)(1 - p_1 - p_2)^2 - \frac{2}{3} (1 - p_1 - p_2)^3 + \frac{1}{2} p_2^2 (2 - p_2) \\
    &\qquad= 2 p_1 (1 - p_1) (1 - p_2) + (1 - 2 p_1) (1 - p_1 - p_2)^2 - \frac{2}{3} (1 - p_1 - p_2)^3 + \frac{1}{2} p_2^2 (2 - p_2), \\
    &\mathbb E[\bar F(X-) (\bar G(\tilde \psi(X-)) + \bar G(\tilde \psi(X-)-)] \\
    &\qquad= p_1 \bar F(0-) (\bar G(\tilde \psi(0-)) + \bar G(\tilde \psi(0-)-)) + (1 - p_1) \int_0^{\tilde x} \bar F(x-) (\bar G(\tilde \psi(x-)) + \bar G(\tilde \psi(x-)-)) \diff \tilde F(x) \\
    &\qquad\qquad + (1 - p_1) \int_{\tilde x}^\infty \bar F(x-) (\bar G(\tilde \psi(x-)) + \bar G(\tilde \psi(x-)-)) \diff \tilde F(x) \\
    &\qquad= 2 (1 - p_1) \int_0^{\tilde x} \bar F(x) F(x) \diff \tilde F(x) + (1 - p_1) (2 - p_2) \int_{\tilde x}^\infty \bar F(x) \diff \tilde F(x) \\
    &\qquad= 2 (1 - p_1) p_1 (1 - p_1 - p_2) + (1 -  2p_1) (1 - p_1 - p_2)^2 - \frac{2}{3} (1 - p_1 - p_2)^3 + \frac{1}{2} (2 - p_2)  p_2^2, \\
    &\mathbb E[\bar G(\tilde \psi(X-)-)] \\
    &\qquad= p_1 \bar G(\tilde \psi(0-)-) + (1 - p_1) \int_0^{\tilde x} \bar G(\tilde \psi(x-)-) \diff \tilde F(x) + (1 - p_1) \int_{\tilde x}^\infty \bar G(\tilde \psi(x-)-) \diff \tilde F(x) \\
    &\qquad= (1 - p_1) \int_0^{\tilde x} F(x) \diff \tilde F(x) + (1 - p_1) \int_{\tilde x}^\infty 1 \diff \tilde F(x) \\
    &\qquad= (1 - p_1) p_1 \int_0^{\tilde x} 1 \diff \tilde F(x) + (1 - p_1)^2 \int_0^{\tilde x} \tilde F(x) \diff \tilde F(x) + (1 - p_1) (1 - \tilde F(\tilde x)) \\
    &\qquad= (1 - p_1) p_1 \tilde F(\tilde x) + \frac{1}{2} (1 - p_1)^2 \tilde F^2(\tilde x) + (1 - p_1) - (1 - p_1) \tilde F(\tilde x) \\
    &\qquad= p_1 (1 - p_1 - p_2) + \frac{1}{2} (1 - p_1 - p_2)^2 + (1 - p_1) - (1 - p_1 - p_2) \\
    &\qquad= p_1 (1 - p_1 - p_2) + \frac{1}{2} (1 - p_1 - p_2)^2 + p_2.
\end{align*}
Now substituting all of these results into Eq.~\eqref{eq:spm_disc_min} yields
\begin{align*}
    \rho_{\min} &= p_1^3 + p_2^3 - 1.
\end{align*}

\end{document}